\documentclass{amsart}%
\usepackage{amsfonts}
\usepackage{amsmath}
\usepackage{amssymb}
\usepackage{graphicx}%
\newtheorem{theorem}{Theorem}
\theoremstyle{plain}

\newtheorem{corollary}{Corollary}

\newtheorem{lemma}{Lemma}

\newtheorem{remark}{Remark}

\numberwithin{equation}{section}
\ifx\pdfoutput\relax\let\pdfoutput=\undefined\fi
\newcount\msipdfoutput
\ifx\pdfoutput\undefined\else
\ifcase\pdfoutput\else
\ifx\paperwidth\undefined\else
\ifdim\paperheight=0pt\relax\else\pdfpageheight\paperheight\fi
\ifdim\paperwidth=0pt\relax\else\pdfpagewidth\paperwidth\fi
\fi\fi\fi
\begin{document}
\title[Non-Archimedean Reaction-Ultradiffusion Equations]{Non-Archimedean Reaction-Ultradiffusion Equations and Complex Hierarchic Systems}
\author{W. A. Z\'{u}\~{n}iga-Galindo}
\address{Centro de Investigaci\'{o}n y de Estudios Avanzados del Instituto
Polit\'{e}cnico Nacional\\
Departamento de Matem\'{a}ticas, Unidad Quer\'{e}taro\\
Libramiento Norponiente \#2000, Fracc. Real de Juriquilla. Santiago de
Quer\'{e}taro, Qro. 76230\\
M\'{e}xico.}
\email{wazuniga@math.cinvestav.edu.mx}
\thanks{The author was partially supported by Conacyt Grant No. 250845.}

\begin{abstract}
We initiate the study of non-Archimedean reaction-ultradiffusion equations and
their connections with models of complex hierarchic systems. From a
mathematical perspective, the equations studied here are the $p$-adic
counterpart of the integro-differential models for phase separation introduced
by Bates and Chmaj. Our equations are also generalizations of the
ultradiffusion equations on trees studied in the 80s by Ogielski, Stein,
Bachas, Huberman, among others, and also generalizations of the master
equations of the Avetisov et al. models, which describe certain complex
hierarchic systems. From a physical perspective, our equations are gradient
flows of non-Archimedean free energy functionals and their solutions describe
the macroscopic density profile of a bistable material whose space of states
has an ultrametric structure. Some of our results are $p$-adic analogs of some
well-known results in the Archimedean \ setting, however, the mechanism of
diffusion is completely different due to the fact that it occurs in an
ultrametric space.

\end{abstract}
\maketitle

\section{Introduction}

In the middle of the 80s the idea of using ultrametric spaces to describe the
states of complex biological systems, which naturally possess a hierarchical
structure, emerged in the works of Frauenfelder, Parisi, Stein, among others,
see e.g. \cite{Dra-Kh-K-V}, \cite{Fraunfelder et al}, \cite{M-P-V},
\cite{R-T-V}. A central paradigm in physics of complex systems (for instance
proteins) asserts that the dynamics of such systems can be modeled as a random
walk in the energy landscape of the system, see e.g. \cite{Fraunfelder et al},
\cite{Kozyrev SV}, and the references therein.\ In protein physics, it is
regarded as one of the most profound ideas put forward to explain the nature
of distinctive life attributes. Typically these landscapes have a huge number
of local minima. It is clear that a description of the dynamics on such
landscapes requires an adequate approximation. The interbasin kinetics method
offers an acceptable solution to this problem. The idea is to study the
kinetics generated \ by transitions between groups of states (basins). In this
framework, the minimal basins correspond to local minima of energy, and the
large basins (superbasins or union of basins) have a hierarchical structure.
By using this approach an energy landscape is approximated by an ultrametric
space (called a disconnectivity graph) and a function on this space describing
the distribution of the activation barriers, see e.g. \cite{Becker et al}. An
ultrametric space $(M,d)$ is a metric space $M$ with a distance satisfying
$d(A,B)\leq\max\left\{  d\left(  A,C\right)  ,d\left(  B,C\right)  \right\}  $
for any three points $A$, $B$, $C$ in $M$. Mezard, Parisi, Sourlas and
Virasoro discovered, in the context of the mean-field theory of spin glasses,
that the space of states of such systems has an ultrametric structure, see
e.g. \cite{M-P-V}, \cite{R-T-V}. The rooted trees have a natural structure of
ultrametric space, and \ the disconnectivity graph above mentioned is an
example of a such space.

After that, a model of hierarchical dynamics is constructed, and by using
\ the postulates of the interbasin kinetics, one gets that the transitions
between basins are described by the following equations:%
\begin{equation}
\frac{\partial u\left(  i,t\right)  }{\partial t}=\sum_{j\neq i}T\left(
j,i\right)  v(j)u\left(  j,t\right)  -\sum_{j\neq i}T\left(  i,j\right)
v\left(  i\right)  u\left(  i,t\right)  ,\text{ }i=1,\ldots,N, \label{Eq_0}%
\end{equation}
where the indices $i$, $j$\ number the states of the system (which correspond
to local minima of energy), $T\left(  i,j\right)  \geq0$\ is the probability
per unit time (or transition rate) of a transition from $i$\ to $j$, and the
$v(j)>0$\ are the basin volumes.\ At this point it is relevant to mention that
equations of type (\ref{Eq_0}) are a generalization of the ultradiffusion
equations on trees studied intensively in the 80s, see e.g. \cite{Bachas-Hu}
and the references therein, and that these equations appeared in models of
protein folding see e.g. \cite{Zwanzig}.

Along this article $p$ will denote a fixed prime number. The field of $p-$adic
numbers $\mathbb{Q}_{p}$\ is defined as the completion of the field of
rational numbers $\mathbb{Q}$\ with respect to the $p-$adic norm $|\cdot|_{p}%
$. See Section 2 for the essential ideas about $p$-adic analysis. Here, we
just mention that in the $p$-adic norm, the integers highly divisible by $p$
are small. For instance, if $p=2$, then $|2^{k}|_{2}=2^{-k}$, for a positive
integer $k$, while for $p=3$, $|2^{k}|_{3}=1$. A such norm satisfies
$|x+y|_{p}\leq\max\left\{  |x|_{p},|y|_{p}\right\}  $, and the metric space
$\left(  \mathbb{Q}_{p},|\cdot|_{p}\right)  $\ is a complete ultrametric
space. This space has a natural hierarchical structure, which is very useful
in physical models involving hierarchies. As a topological space
$\mathbb{Q}_{p}$\ is homeomorphic to a Cantor-like subset of the real line,
i.e. $\mathbb{Q}_{p}$\ is a fractal. The $p-$adic norm can be extended to
$\mathbb{Q}_{p}^{n}$\ by taking for $x=(x_{1},\cdots,x_{n})$, $||x||_{p}%
:=\max_{i}|x_{i}|_{p}$.

Around 2000, Avetisov et al. discovered, among several things, that under
suitable physical and mathematical hypotheses, the ultradiffusion equations on
trees studied by Ogielski, Stain, Bachas and Huberman, among several others,
see e.g. \cite{Bachas-Hu}, have a `continuous $p$-adic limit'. We explain
briefly these ideas following \cite{O-S} and \cite{Av-8}. The states of system
are labeled by the numbers \ $0$, $1$, \ldots,$2^{n-1}-1$, for some $n$. By
representing each number in base $2$, we get a set of $2^{n}$ binary vectors
of lenght $n$. This set forms a rooted tree with $n+1$ levels and $2^{n}$
branches (states). \ An ultrametric distance $d$ between two branches is given
by the number of levels that it is necessary to descend up to the branches
merge. Now, a stochastic dynamical system on the rooted tree is introduced by
using a random walk (on the top level) to model the transition between states.
\ We denote by $P_{i}(t)$ the probability of occupying the state $i$ at the
time $t$, and set $\boldsymbol{P}(t)=\left[  P_{0}(t),\ldots,P_{2^{n}%
-1}(t)\right]  ^{T}$. We now assume that the probability per unit of time of
jumping from state $i$ into state $j$ is a function \ of the ultrametric
distance, \ in this way we construct a matrix $Q=\left[  f\left(  d\left(
i,j\right)  \right)  \right]  $, the Parisi matrix of the system, and the
dynamics of the system is controlled by \ the master equation
\begin{equation}
\frac{d\boldsymbol{P}(t)}{dt}=Q\boldsymbol{P}(t).\label{master_eq_O_S}%
\end{equation}
See \cite{O-S} for further details. On the other hand, in the models of
spontaneous breaking of the replica symmetry are used for the investigation of
spin glasses, a replica matric $Q=\left[  Q_{ab}\right]  $ of size $n\times n$
\ occurs. This matrix is constructed as follows: consider the set of integers
$m_{i}$, with $i\in\left\{  1,\ldots,N\right\}  $, where $\frac{m_{i}}%
{m_{i-1}}$ are integers for $i>1$ and $\frac{n}{m_{i}}$ are also integers. The
entries of the replica matrix are defined as follows:%
\[
Q_{aa}=0\text{, and for }a\neq b\text{,\ }Q_{ab}=q_{i}\text{, with }\left[
\frac{a}{m_{i}}\right]  \neq\left[  \frac{b}{m_{i}}\right]  \text{
and\ }\left[  \frac{a}{m_{i+1}}\right]  =\left[  \frac{b}{m_{i+1}}\right]  ,
\]
here $\left[  \cdot\right]  $ denotes the integer part function. To obtain a
$p$-adic parametrization of the replica matrix (the Parisi matrix), we use the
set $\left\{  1,\ldots,p^{N}\right\}  $ (instead of $\left\{  1,\ldots
,N\right\}  $) and the mapping $l:\left\{  1,\ldots,p^{N}\right\}  \rightarrow
p^{-N}\mathbb{Z}/\mathbb{Z}$ defined as $l^{-1}\left(  \sum_{j=1}^{N}%
x_{j}p^{-j}\right)  \allowbreak=1+p^{-1}\sum_{j=1}^{N}x_{j}p^{j}$, where
$x_{j}$ are $p$-adic digits. In \cite{Av-8}, Avetisov et al. established that
the replica matrix can be parametrized as $Q_{ab}=f\left(  \left\vert
l(a)-l(b)\right\vert _{p}\right)  $ where $f$ \ is a function such that
$f(p^{i})=q_{i\text{ }}$ and $f(0)=0$, and that in `the limit $N$ tends to
infinity,' master equation (\ref{master_eq_O_S}) becomes the $p$-adic heat
equation:
\begin{equation}
\frac{\partial u\left(  x,t\right)  }{\partial t}+D_{x}^{\alpha}u\left(
x,t\right)  =0\text{, \ }x\in\mathbb{Q}_{p}\text{, }t\geq
0,\label{heat_equation}%
\end{equation}
where $D_{x}^{\alpha}$ is the Vladimirov operator, and the matrix $\left[
Q_{ab}\right]  $ represents this operator in a convenient finite dimensional
space, i.e. $\left[  Q_{ab}\right]  $ is a discretization of the Vladimirov
operator. For further details, the reader may consult \cite{Av-8},
\cite{KKZuniga}, \cite{Zuniga-LNM-2016}. \ The terminology `$p$-adic heat
equation' means that a Markov process is attached to the fundamental solution
of (\ref{heat_equation}), consequently, we can consider \ (\ref{heat_equation}%
) as a $p$-adic analogue of the classical heat equation. For an in-depth
discussion of this analogy, the reader may consult \cite{KKZuniga},
\cite{Zuniga-LNM-2016} and the references therein. On the other hand, many
different ultrametrics can be defined on a rooted tree, at first sight, the
above limit process works only for ultrametrics of the form $f\circ\left\vert
\cdot\right\vert _{p}$, however, this matter has not been investigated yet.

We now come back to equation (\ref{Eq_0}) to explain how the above discussed
ideas fit into it. We rewrite (\ref{Eq_0}) as follows:\
\begin{equation}
\frac{\partial}{\partial t}\left[
\begin{array}
[c]{c}%
u\left(  1,t\right) \\
\vdots\\
u\left(  j,t\right) \\
\vdots\\
u\left(  N,t\right)
\end{array}
\right]  =\left[  W(j,i)-\delta_{ii}W\left(  i,i\right)  \right]  _{N\times
N}\left[
\begin{array}
[c]{c}%
u\left(  1,t\right) \\
\vdots\\
u\left(  j,t\right) \\
\vdots\\
u\left(  N,t\right)
\end{array}
\right]  , \label{Eq_0_A}%
\end{equation}
where $W(j,i)=T(j,i)v(j)$, for $j\neq i$, and $W\left(  i,i\right)  =v\left(
i\right)  \sum_{j\neq i}T\left(  i,j\right)  $. We assume that the space of
states $\left\{  1,\ldots,N\right\}  $ have a hierarchical structure, then
$\left[  W(j,i)-\delta_{ii}W\left(  i,i\right)  \right]  $ is a Parisi-type
matrix. These matrices appear naturally in models of complex systems such as
spin glasses, see e.g. \cite{Bachas-Hu}, \cite{M-P-V}, \cite{O-S},
\cite{R-T-V}, \cite{KKZuniga} and the references therein. Under mild
hypotheses, these Parisi matrices can be parametrized by $p$-adic numbers, and
then the master equation (\ref{Eq_0_A}) becomes a discretization of $p$-adic
ultradiffusion equation of type (\ref{Eq_1}). Consequently, it is completely
natural to propose that in the limit when $N$ tends to infinity the master
equation (\ref{Eq_0_A}) becomes a $p$-adic ultradiffusion equation, see e.g.
\cite[Chapters 4 ,8]{KKZuniga}, \cite[Chapter 2, 3]{Zuniga-LNM-2016} and the
references therein.

The `$p$-adic limit' of master equations (\ref{Eq_0_A}) have the form:%
\begin{equation}
\frac{\partial u\left(  x,t\right)  }{\partial t}=\int_{\mathbb{Q}_{p}^{n}%
}J\left(  \left\Vert x-y\right\Vert _{p}\right)  \left[  u\left(  y,t\right)
-u\left(  x,t\right)  \right]  d^{n}y, \label{Eq_1}%
\end{equation}
$x\in\mathbb{Q}_{p}^{n},t\geq0$. The function $u(x,t):\mathbb{Q}_{p}^{n}%
\times\mathbb{R}_{+}\rightarrow\mathbb{R}_{+}$\ is a probability density
distribution, so that $\int_{B}u\left(  x,t\right)  d^{n}x$\ is the
probability of finding the system in a domain $B\subset\mathbb{Q}_{p}^{n}$\ at
the instant $t$. The function $J\left(  \left\Vert x-y\right\Vert _{p}\right)
:\mathbb{Q}_{p}^{n}\times\mathbb{Q}_{p}^{n}\rightarrow\mathbb{R}_{+}$\ is the
probability of the transition \ per unit of time (or transition rate) from
state $y$\ to state $x$.\ It is known that for many $J$s, equations of type
(\ref{Eq_1}) are ultradiffusion equations i.e. they are $p$-adic counterparts
of the classical heat equations. For instance, $J(\left\Vert x\right\Vert
_{p})=\left\Vert x\right\Vert _{p}^{\gamma}e^{-\left\Vert x\right\Vert _{p}}$,
with $\gamma>-n$, \ corresponds to an exponential landscape in the sense of
\cite{Av-4}, in this case, the fundamental solution of (\ref{Eq_1}) is the
transition density of a bounded right-continuous Markov process without second
kind discontinuities, see \cite{TB-zuniga}, \cite{Bendikov} and the references
therein. As a consequence of the work of many people, among them, Vlamimirov,
Volovich, Zelenov, Avetisov, Kozyrev, Kochubei, Khrennikov, Albeverio, and
Z\'{u}\~{n}iga-Galindo, we have now a good theory of $p$-adic `linear'
reaction-diffusion equations which has emerged motivated by connections
between $p$-adic analysis and models of complex systems. For further details
the reader may consult \cite{Zuniga-LNM-2016}, \cite{KKZuniga}, see also
\cite{TB-zuniga}, as well as the classics, \cite{V-V-Z}, \cite{Koch}.

In our opinion, the novelty and relevance of the `idealistic models' of
Avetisov et al. come from two facts: first, they codify, in a mathematical
language, the central physical paradigm asserting that the dynamics of (many)
complex systems can be described as a random walk on an ultrametric space;
second, these models give a description of the characteristic types of
relaxation of complex systems. The original models of Avetisov et al. were
formulated in dimension one, more precisely, these models were constructed by
using `exactly one' cross section of an energy landscape. In \cite[p. 98 and
figures 11.3 and 11.4]{Fraunfelder et al} Frauenfelder et al. have pointed out
that using `one' cross section of an energy landscape of a complex systems to
describe its dynamics is misleading, because it appears that the transition
from an initial state to a final state must follow a unique pathway, and
entropy does not play a role. By considering several trees and by using the
above mentioned limit process, one gets $n$-dimensional $p$-adic master
equations of type (\ref{Eq_1}). The Frauenfelder et al. observation and the
Avetisov et al. work provide a strongly motivation for developing a general
theory of $n-$dimensional $p$-adic reaction-ultradiffusion equations.

From the perspectives of mathematics and physics, a natural step in the
investigation of equations of types (\ref{Eq_0_A})-(\ref{Eq_1}) is to
introduce a non-linear reaction term. This article aims to initiate the theory
of $n-$dimensional $p$-adic reaction-ultradiffusion equations, and their
connections with models of complex hierarchic systems. The terminology
`reaction-diffusion equations' has been used in connection with the models of
Avetisov et al., see e.g. \cite{Bikulov}, to mean (linear) parabolic-type
equations with variable coefficients. A general theory for this type of
equations is given in \cite{Zuniga-LNM-2016}, see also \cite{KKZuniga},
\cite{Koch}. Here `reaction-diffusion equations' means nonlinear equations,
such as is commonly used in the Euclidean case, see e.g. \cite{Fife},
\cite{Gind}, \cite{smaller}. We use the term ultradiffusion instead of
diffusion, due to fact that in classical probability the term diffusion is
used only in connection with stochastic processes with continuous paths, and
in the $p$-adic setting, the paths cannot be continuous.

We study equations of the type%
\begin{equation}
\frac{\partial u\left(  x,t\right)  }{\partial t}=\int_{\mathbb{Q}_{p}^{n}%
}J\left(  \left\Vert x-y\right\Vert _{p}\right)  \left[  u\left(  y,t\right)
-u\left(  x,t\right)  \right]  d^{n}y-\lambda f\left(  u\left(  x,t\right)
\right)  , \label{Eq_2}%
\end{equation}
where $J\left(  \left\Vert x\right\Vert _{p}\right)  \geq0$, $\int
_{\mathbb{Q}_{p}^{n}}J\left(  \left\Vert x\right\Vert _{p}\right)  d^{n}x=1$,
$\lambda>0$ sufficiently large and $f$ is (for instance) a polynomial having
roots in $-1$, $0$, $1$. Formally, equation (\ref{Eq_2}) is the $L^{2}%
$-gradient flow of the following non-Archimedean Helmholtz free-energy
functional:%
\begin{equation}
E\left[  \varphi\right]  =\frac{1}{4}\int_{\mathbb{Q}_{p}^{n}}\int
_{\mathbb{Q}_{p}^{n}}J\left(  \left\Vert x-y\right\Vert _{p}\right)  \left\{
\varphi\left(  x\right)  -\varphi\left(  y\right)  \right\}  ^{2}d^{n}%
xd^{n}y+\lambda\int_{\mathbb{Q}_{p}^{n}}W\left(  \varphi\left(  x\right)
\right)  d^{n}x, \label{Eq_3}%
\end{equation}
where $\varphi$ is a function taking values in the interval $\left[
-1,1\right]  $ and $W$ is a double-well potential.

Equations of the form (\ref{Eq_2}) can be well-approximated in finite
dimensional real spaces by ODE's. In a suitable basis, where the unknown
function is identified with the column vector $\left[  u\left(  \boldsymbol{i}%
,t\right)  \right]  _{\boldsymbol{i}\in G_{N}^{n}}$, these equations have the
form
\begin{equation}
\frac{\partial}{\partial t}\left[  u\left(  \boldsymbol{i},t\right)  \right]
_{\boldsymbol{i}\in G_{N}^{n}}=-A^{\left(  N\right)  }\left[  u\left(
\boldsymbol{i},t\right)  \right]  _{\boldsymbol{i}\in G_{N}^{n}}%
-\lambda\left[  f\left(  u\left(  \boldsymbol{i},t\right)  \right)  \right]
_{\boldsymbol{i}\in G_{N}^{n}}, \label{Eq_4}%
\end{equation}
where $A^{\left(  N\right)  }$ is the matrix representation of a linear
operator that approximates, in a suitable finite dimensional vector space, the
integral operator involving the function $J$ in the right-side of
(\ref{Eq_2}). Equation (\ref{Eq_4}) is $L^{2}$-gradient flow of a `finite'
Helmholtz energy functional, i.e. a functional defined on the space $G_{N}%
^{n}$. In Section \ref{Sect_G_Landau},\ we present some results about the
convergence of finite Helmholtz functionals when $N$ tends to infinity.
Equations of type (\ref{Eq_4}) are generalizations of \ ultradiffusion
equations on trees considered in \cite{Bachas-Hu}. The set $G_{N}^{n}$ is a
finite ultrametric space, this class of spaces contains as \ particular case
the finite rooted trees.

This article is dedicated to study the interplay between all the above
mentioned objects and their physical significance. We determine the spaces and
conditions for which the Cauchy problems for equations (\ref{Eq_2}%
)-(\ref{Eq_4}) are well-posed, see Theorems \ref{Theorem3}, \ref{Theorem3A}.
We show that equations (\ref{Eq_2})-(\ref{Eq_4}) have stationary solutions
with `arbitrary interfaces', this means, in the case of equation (\ref{Eq_2}),
the following.\ Given a ball $B_{N_{0}}^{n}\left(  x_{0}\right)  $ of radius
$p^{N_{0}}$ centered at $x_{0}$, $\mathbb{Q}_{p}^{n}$ can be divided into
three disjoint sets $M$, $B_{N_{0}}^{n}\left(  x_{0}\right)  \smallsetminus
M$, and\ $\mathbb{Q}_{p}^{n}\smallsetminus B_{N_{0}}^{n}\left(  x_{0}\right)
$. Equation (\ref{Eq_2}) admits a stationary solution $\widetilde{u}\left(
x\right)  $ satisfying $\alpha^{+}\leq\widetilde{u}\left(  x\right)  \leq1$
for $x\in M$,\ $-1\leq\widetilde{u}\left(  x\right)  \leq\alpha^{-}$ for $x\in
B_{N_{0}}^{n}\left(  x_{0}\right)  \smallsetminus M$, and $\lim_{\left\Vert
x\right\Vert _{p}\rightarrow\infty}\widetilde{u}\left(  x\right)  =0$, for
some suitable constants $\alpha^{+}$, $\alpha^{-}$, see Theorems
\ref{Theorem1}, \ref{Theorem1A}.

We also show that the solution of Cauchy problem attached to (\ref{Eq_4})
converges to the solution of the Cauchy problem attached to (\ref{Eq_2}), when
$N$ tends to infinity, in the case in which the initial condition for equation
(\ref{Eq_4}) is a continuous function taking values in the interval $\left[
-1,1\right]  $, see Theorem \ref{Theorem4}. Roughly speaking, equation
(\ref{Eq_2}) is the `$p$-adic continuous limit' of the system of equations
(\ref{Eq_4}), when $N$ tends to infinity. The matrix $A^{\left(  N\right)  }$
in equation (\ref{Eq_4}) is the $Q$-matrix of a finite homogeneous Markov
chain with state space $G_{N}^{n}$, and equation (\ref{Eq_4}) with $f=0$ is
the Kolmogorov backward equation attached to this Markov chain, see Theorem
\ref{Theorem2B}.

From a physical perspective equations (\ref{Eq_2})-(\ref{Eq_4}) model phase
separation of bistable materials whose space of states have an ultrametric
structure. Our models are the $p$-adic counterparts of the
integro-differential models for phase separation due to Bates and Chmaj, see
\cite{Bates-Fife}-\cite{BatesChmdj}, and \cite{Alberti-B-1}-\cite{Alberti-B-2}%
. The function $u(x,t)$, respectively $\left[  u\left(  i,t\right)  \right]
_{\boldsymbol{i}\in G_{N}^{n}}$, the order parameter, represents the
macroscopic density profile of a material, which has two equilibrium states
$u(x,t)\equiv-1$, $u(x,t)\equiv1$, and $-1<u(x,t)<1$ represents the
`interface', and equations (\ref{Eq_2})-(\ref{Eq_4}) model a transition
between the equilibrium phases. Theorems \ref{Theorem1}, \ref{Theorem1A} show
that our models of bistable systems can develop arbitrary stable interfaces.

In the proofs of our results we have used freely techniques of abstract
evolution equations, for instance \cite{C-H}, \cite{Milan}, and adapted
techniques and ideas of the classical reaction-diffusion equations, for
instance \cite{Gind}, \cite{smaller}. However, the non-Archimedean theory is
far from being a straightforward consequence of the classical theory of
reaction-diffusion equations. For instance, the existence of traveling waves,
that usually emerge in the transformation between the pure phases, is an open
problem in the non-Archimedean case, since the classical ideas cannot be
applied directly. On the other hand, the non-Archimedean comparison theorem
needed here, see Theorem \ref{Theorem2} and Corollary \ref{Cor2A}, requires
the condition that the `volume of the system be sufficiently large'.

In a recent book by Dellacherie, Martinez and San Martin, see \cite{D-Ma-SM},
the authors present a theory of `ultrametric matrices' and their connections
with Markov chains. As far as we understand, these ultrametric matrices are
generalizations of the inverses of Parisi matrices. In particular, the results
of this book imply that our Theorem \ref{Theorem2B} is valid for more general
type of matrices. It is interesting to mention that Khrennikov and Kozyrev
developed a very general theory of Parisi-type matrices, see \cite{Khr-Koz4}%
-\cite{Khr-Koz6}.

Finally, our Theorem \ref{Theorem4} allows us to produce numerical simulations
of the behavior of the bistable systems whose states are described by the
solutions of our $p$-adic reaction-ultradiffusion equations. We prefer focus
our article on mathematical aspects, however, our equations include as a
particular case the reaction-diffusion equations on graphs studied by Ueyama
and Hosoe in \cite{Ueyama}. The numerical simulations presented in this
article show that our $p$-adic reaction-ultradiffusion equations develop
stable patterns such as occurs in the Archimedean case. Here, it is important
to mention that `essentially' there is no available literature on numerical
methods for nonlinear $p$-adic reaction-ultradiffusion equations. For some
special equations, numerical solutions can be obtained using $p$-adic
wavelets, see \cite{KKZuniga}, \cite{Kozyrev SV}, \ and the references
therein. In our opinion, this technique is not applicable to the type of
equations considered here. There are well-known mathematical techniques for
the discretization abstract nonlinear evolution equations, we use some of them
here, but the challenge is the `visualization of the data.' Then, the
numerical study of the stable patterns of the equations introduced here, and
the comparison with the stable patterns corresponding to the classical
reaction-diffusion equations is an open problem.

\section{$p$\textbf{-}Adic Analysis: Essential Ideas}

\subsection{The field of $p$-adic numbers}

Along this article $p$ will denote a prime number. The field of $p-$adic
numbers $\mathbb{Q}_{p}$ is defined as the completion of the field of rational
numbers $\mathbb{Q}$ with respect to the $p-$adic norm $|\cdot|_{p} $, which
is defined as%
\[
\left\vert x\right\vert _{p}=\left\{
\begin{array}
[c]{lll}%
0 & \text{if} & x=0\\
p^{-\gamma} & \text{if} & x=p^{\gamma}\frac{a}{b}\text{,}%
\end{array}
\right.
\]
where $a$ and $b$ are integers coprime with $p$. The integer $\gamma:=ord(x)
$, with $ord(0):=+\infty$, is called the\textit{\ }$p-$\textit{adic order of}
$x$. We extend the $p-$adic norm to $\mathbb{Q}_{p}^{n}$ by taking%
\[
||x||_{p}:=\max_{1\leq i\leq n}|x_{i}|_{p},\qquad\text{for }x=(x_{1}%
,\dots,x_{n})\in\mathbb{Q}_{p}^{n}.
\]
We define $ord(x)=\min_{1\leq i\leq n}\{ord(x_{i})\}$, then $||x||_{p}%
=p^{-ord(x)}$.\ The metric space $\left(  \mathbb{Q}_{p}^{n},||\cdot
||_{p}\right)  $ is a complete ultrametric space. As a topological space
$\mathbb{Q}_{p}$\ is homeomorphic to a Cantor-like subset of the real line,
see e.g. \cite{Alberio et al}, \cite{V-V-Z}.

Any $p-$adic number $x\neq0$ has a unique expansion of the form
\[
x=p^{ord(x)}\sum_{j=0}^{\infty}x_{i}p^{j},
\]
where $x_{j}\in\{0,1,2,\dots,p-1\}$ and $x_{0}\neq0$. In addition, any $p-
$adic number $x\neq0$ can be represented uniquely as $x=p^{ord(x)}ac\left(
x\right)  $ where $ac\left(  x\right)  =\sum_{j=0}^{\infty}x_{i}p^{j}$,
$x_{0}\neq0$, is called the \textit{angular component} of $x$. Notice that
$\left\vert ac\left(  x\right)  \right\vert _{p}=1$.

\subsection{Topology of $\mathbb{Q}_{p}^{n}$}

For $r\in\mathbb{Z}$, denote by $B_{r}^{n}(a)=\{x\in\mathbb{Q}_{p}%
^{n};||x-a||_{p}\leq p^{r}\}$ \textit{the ball of radius }$p^{r}$ \textit{with
center at} $a=(a_{1},\dots,a_{n})\in\mathbb{Q}_{p}^{n}$, and take $B_{r}%
^{n}(0):=B_{r}^{n}$. Note that $B_{r}^{n}(a)=B_{r}(a_{1})\times\cdots\times
B_{r}(a_{n})$, where $B_{r}^{1}(a_{i}):=B_{r}(a_{i})=\{x\in\mathbb{Q}%
_{p};|x_{i}-a_{i}|_{p}\leq p^{r}\}$ is the one-dimensional ball of radius
$p^{r}$ with center at $a_{i}\in\mathbb{Q}_{p}$. The ball $B_{0}^{n}$ equals
the product of $n$ copies of $B_{0}=\mathbb{Z}_{p}$, \textit{the ring of }%
$p-$\textit{adic integers}. We also denote by $S_{r}^{n}(a)=\{x\in
\mathbb{Q}_{p}^{n};||x-a||_{p}=p^{r}\}$ \textit{the sphere of radius }$p^{r}$
\textit{with center at} $a=(a_{1},\dots,a_{n})\in\mathbb{Q}_{p}^{n}$, and take
$S_{r}^{n}(0):=S_{r}^{n}$. We notice that $S_{0}^{1}=\mathbb{Z}_{p}^{\times}$
(the group of units of $\mathbb{Z}_{p}$), but $\left(  \mathbb{Z}_{p}^{\times
}\right)  ^{n}\subsetneq S_{0}^{n}$. The balls and spheres are both open and
closed subsets in $\mathbb{Q}_{p}^{n}$. In addition, two balls in
$\mathbb{Q}_{p}^{n}$ are either disjoint or one is contained in the other.

As a topological space $\left(  \mathbb{Q}_{p}^{n},||\cdot||_{p}\right)  $ is
totally disconnected, i.e. the only connected \ subsets of $\mathbb{Q}_{p}%
^{n}$ are the empty set and the points. A subset of $\mathbb{Q}_{p}^{n}$ is
compact if and only if it is closed and bounded in $\mathbb{Q}_{p}^{n}$, see
e.g. \cite[Section 1.3]{V-V-Z}, or \cite[Section 1.8]{Alberio et al}. The
balls and spheres are compact subsets. Thus $\left(  \mathbb{Q}_{p}%
^{n},||\cdot||_{p}\right)  $ is a locally compact topological space.

We will use $\Omega\left(  p^{-r}||x-a||_{p}\right)  $ to denote the
characteristic function of the ball $B_{r}^{n}(a)$. For more general sets, we
will use the notation $1_{A}$ for the characteristic function of a set $A$.

\subsection{The Bruhat-Schwartz space}

A real-valued function $\varphi$ defined on $\mathbb{Q}_{p}^{n}$ is
\textit{called locally constant} if for any $x\in\mathbb{Q}_{p}^{n}$ there
exists an integer $l(x)\in\mathbb{Z}$ such that%
\begin{equation}
\varphi(x+x^{\prime})=\varphi(x)\text{ for }x^{\prime}\in B_{l(x)}^{n}.
\label{local_constancy_0}%
\end{equation}

A function $\varphi:\mathbb{Q}_{p}^{n}\rightarrow\mathbb{R}$ is called a
\textit{Bruhat-Schwartz function (or a test function)} if it is locally
constant with compact support. Any test function can be represented as a
linear combination, with real coefficients, of characteristic functions of
balls. The $\mathbb{R}$-vector space of Bruhat-Schwartz functions is denoted
by $\mathcal{D}(\mathbb{Q}_{p}^{n})$. For $\varphi\in\mathcal{D}%
(\mathbb{Q}_{p}^{n})$, the largest number $l=l(\varphi)$ satisfying
(\ref{local_constancy_0}) is called \textit{the exponent of local constancy
(or the parameter of constancy) of} $\varphi$. Let denote by $\mathcal{D}%
_{N}^{l}\left(
%TCIMACRO{\U{211a} }%
%BeginExpansion
\mathbb{Q}
%EndExpansion
_{p}^{n}\right)  $ (or simply $\mathcal{D}_{N}^{l}$), the finite dimensional
subspace of all real-valued test functions having supports in $B_{N}^{n}$ and
with parameters of constancy $\geq l$. Then $\mathcal{D}_{N}^{l}%
\subset\mathcal{D}_{N^{\prime}}^{l^{\prime}}$ if $N^{\prime}\geq N$ and $l\geq
l^{\prime}$.

If $U$ is an open subset of $\mathbb{Q}_{p}^{n}$, $\mathcal{D}(U)$ denotes the
space of test functions with supports contained in $U$, then $\mathcal{D}(U)$
is dense in
\[
L^{\rho}\left(  U\right)  =\left\{  \varphi:U\rightarrow\mathbb{R};\left(
\int_{\mathbb{Q}_{p}^{n}}\left\vert \varphi\left(  x\right)  \right\vert
^{\rho}d^{n}x\right)  ^{\frac{1}{\rho}}<\infty\right\}  ,
\]
where $d^{n}x$ is the Haar measure on $\mathbb{Q}_{p}^{n}$ normalized by the
condition $vol(B_{0}^{n})\allowbreak=1$, for $1\leq\rho<\infty$, see e.g.
\cite[Section 4.3]{Alberio et al}. In the case $U=\mathbb{Q}_{p}^{n}$, we will
use the notation $L^{\rho}$ instead of $L^{\rho}\left(  \mathbb{Q}_{p}%
^{n}\right)  $. For an in depth discussion about $p$-adic analysis the reader
may consult \ \cite{Alberio et al}, \cite{Koch}, \cite{Taibleson},
\cite{V-V-Z}.

\section{Some Functional Spaces and Operators}

We define $X_{\infty}(%
%TCIMACRO{\U{211a} }%
%BeginExpansion
\mathbb{Q}
%EndExpansion
_{p}^{n}):=X_{\infty}=\overline{\left(  \mathcal{D}(%
%TCIMACRO{\U{211a} }%
%BeginExpansion
\mathbb{Q}
%EndExpansion
_{p}^{n}),\left\Vert \cdot\right\Vert _{\infty}\right)  }$, where $\left\Vert
\phi\right\Vert _{\infty}=\sup_{x\in%
%TCIMACRO{\U{211a} }%
%BeginExpansion
\mathbb{Q}
%EndExpansion
_{p}^{n}}|\phi(x)|$ and the bar means the completion with respect the metric
induced by $\left\Vert \cdot\right\Vert _{\infty}$. We also use $\left\Vert
\cdot\right\Vert _{\infty}$ to denote the extension of $\left\Vert
\cdot\right\Vert _{\infty}$ to $X_{\infty}$. Notice that all the functions in
$X_{\infty}$ are continuous and \ that
\[
X_{\infty}\subset C_{0}:=\left(  \left\{  f:%
%TCIMACRO{\U{211a} }%
%BeginExpansion
\mathbb{Q}
%EndExpansion
_{p}^{n}\rightarrow\mathbb{R};f\text{ continuous with }\lim_{\left\Vert
x\right\Vert _{p}\rightarrow\infty}f\left(  x\right)  =0\right\}  ,\left\Vert
\cdot\right\Vert _{\infty}\right)  .
\]
On the other hand, since $\mathcal{D}(%
%TCIMACRO{\U{211a} }%
%BeginExpansion
\mathbb{Q}
%EndExpansion
_{p}^{n})$ is dense in $C_{0}$, cf. \cite[Chap. II, Proposition 1.3]%
{Taibleson}, we conclude that $X_{\infty}=C_{0}$. In a more general case, if
$K$ is an open subset of $%
%TCIMACRO{\U{211a} }%
%BeginExpansion
\mathbb{Q}
%EndExpansion
_{p}^{n}$, we define $X_{\infty}\left(  K\right)  =\overline{\left(
\mathcal{D}(K),\left\Vert \cdot\right\Vert _{\infty}\right)  }$.

We set
\[
\text{ }X_{N}:=\left(  \mathcal{D}_{N}^{-N}\left(
%TCIMACRO{\U{211a} }%
%BeginExpansion
\mathbb{Q}
%EndExpansion
_{p}^{n}\right)  ,\left\Vert \cdot\right\Vert _{\infty}\right)  \text{ for
}N\geq1\text{.}%
\]
Any $\varphi\in X_{N}$ has support in $B_{N}^{n}=\left(  p^{-N}\mathbb{Z}%
_{p}\right)  ^{n}$, and \ $\varphi$ satisfies (\ref{local_constancy_0}) for
$x^{\prime}\in B_{-N}^{n}=\left(  p^{N}\mathbb{Z}_{p}\right)  ^{n}$, in
addition, $B_{\pm N}^{n}$ are additive subgroups and $G_{N}^{n}:=B_{N}%
^{n}/B_{-N}^{n}$ is a finite group with $\#G_{N}^{n}:=p^{2Nn}$ elements. Any
element $\boldsymbol{i}=(\boldsymbol{i}_{1},\ldots,\boldsymbol{i}_{n})$ of
$G_{N}^{n}$ can be represented as
\begin{equation}
\boldsymbol{i}_{j}=a_{-N}^{j}p^{-N}+a_{-N+1}^{j}p^{-N+1}+\ldots+a_{0}%
^{j}+a_{1}^{j}p+\ldots+a_{N-1}^{j}p^{N-1}\text{ } \label{representatives}%
\end{equation}
for $j=1,\ldots,n$, with $a_{k}^{j}\in\left\{  0,1,\ldots,p-1\right\}  $. From
now on, we fix a set of representatives in $%
%TCIMACRO{\U{211a} }%
%BeginExpansion
\mathbb{Q}
%EndExpansion
_{p}^{n}$ for $G_{N}^{n}$ of the form (\ref{representatives}). We denote by
$\Omega\left(  p^{M}\left\Vert x-x_{0}\right\Vert _{p}\right)  $, the
characteristic function of the ball $x_{0}+\left(  p^{M}\mathbb{Z}_{p}\right)
^{n}$. We notice that any non-zero function $\varphi$ in $X_{N}$ has an index
of local constancy $l_{\varphi}\in\left\{  -N,-N+1,\ldots,0,1,\ldots
,N\right\}  $, and that $B_{l_{\varphi}}^{n}$can be covered by a finite
disjoint union of balls of the form $B_{-N}^{n}\left(  \boldsymbol{j}\right)
$, with $\boldsymbol{j}\in G_{N}^{n}$, then $\left\{  \Omega\left(
p^{N}\left\Vert x-\boldsymbol{i}\right\Vert _{p}\right)  \right\}
_{\boldsymbol{i}\in G_{N}^{n}}$ is a basis of $\mathcal{D}_{N}^{-N}$, see also
e.g. \cite[Lemma 4.3.1]{Alberio et al}. We notice now that if $\varphi\left(
x\right)  =\sum_{\boldsymbol{i}\in G_{N}^{n}}\varphi\left(  \boldsymbol{i}%
\right)  \Omega\left(  p^{N}\left\Vert x-\boldsymbol{i}\right\Vert
_{p}\right)  $, \ with $\varphi\left(  \boldsymbol{i}\right)  \in\mathbb{R}$,
then $\left\Vert \varphi\right\Vert _{\infty}=\max_{\boldsymbol{i}}\left\vert
\varphi\left(  \boldsymbol{i}\right)  \right\vert $. Hence $X_{N}$ is
isomorphic as a Banach space to $\left(  \mathbb{R}^{\#G_{N}^{n}},\left\Vert
\cdot\right\Vert _{\mathbb{R}}\right)  $, where $\left\Vert \left(
t_{1},\ldots,t_{\#G_{N}^{n}}\right)  \right\Vert _{\mathbb{R}}=\max_{1\leq
j\leq\#G_{N}^{n}}\left\vert t_{j}\right\vert $.

We now define for $N\geq1$, $P_{N}:X_{\infty}\rightarrow X_{N}$ as
\[
P_{N}\varphi\left(  x\right)  =\sum_{\boldsymbol{i}\in G_{N}^{n}}%
\varphi\left(  \boldsymbol{i}\right)  \Omega\left(  p^{N}\left\Vert
x-\boldsymbol{i}\right\Vert _{p}\right)  .
\]
Therefore $P_{N}$ is a linear bounded operator, indeed, $\left\Vert
P_{N}\right\Vert \leq1$.

\begin{lemma}
\label{Lemma_0}$\lim_{N\rightarrow\infty}\left\Vert \varphi-P_{N}%
\varphi\right\Vert _{\infty}=0$ for any $\varphi\in X_{\infty}$.
\end{lemma}

\begin{proof}
By using the fact that $\mathcal{D}(%
%TCIMACRO{\U{211a} }%
%BeginExpansion
\mathbb{Q}
%EndExpansion
_{p}^{n})$ is dense in $X_{\infty}$, given any $\epsilon$ sufficiently small,
there exists $\widetilde{\varphi}\in\mathcal{D}_{M}^{l}\left(
%TCIMACRO{\U{211a} }%
%BeginExpansion
\mathbb{Q}
%EndExpansion
_{p}^{n}\right)  $, with $l$, $M$ integers depending on $\epsilon$, such that
$\left\Vert \varphi-\widetilde{\varphi}\right\Vert _{\infty}<\epsilon$. We may
assume without loss of generality that $M\geq1$ since $\mathcal{D}_{M}%
^{l}\subset\mathcal{D}_{M+1}^{l}$, and that $l<0$ since\ if $l\geq0$ then
$\mathcal{D}_{M}^{l}\subset\mathcal{D}_{M}^{-l}$. Thus $\widetilde{\varphi}%
\in\mathcal{D}_{M}^{-k}\left(
%TCIMACRO{\U{211a} }%
%BeginExpansion
\mathbb{Q}
%EndExpansion
_{p}^{n}\right)  $ with $M$, $k\geq1$, and $\mathcal{D}_{M}^{-k}\left(
%TCIMACRO{\U{211a} }%
%BeginExpansion
\mathbb{Q}
%EndExpansion
_{p}^{n}\right)  \subset\mathcal{D}_{\max\left\{  k,M\right\}  }%
^{-\max\left\{  k,M\right\}  }\left(
%TCIMACRO{\U{211a} }%
%BeginExpansion
\mathbb{Q}
%EndExpansion
_{p}^{n}\right)  \subset\mathcal{D}_{N}^{-N}\left(
%TCIMACRO{\U{211a} }%
%BeginExpansion
\mathbb{Q}
%EndExpansion
_{p}^{n}\right)  $ for $N\geq\max\left\{  k,M\right\}  $, i.e. $\widetilde
{\varphi}\in\mathcal{D}_{N}^{-N}\left(
%TCIMACRO{\U{211a} }%
%BeginExpansion
\mathbb{Q}
%EndExpansion
_{p}^{n}\right)  $ for $N\geq\max\left\{  k,M\right\}  $, and $\widetilde
{\varphi}\left(  x\right)  =\sum_{\boldsymbol{i}\in G_{N}^{n}}\widetilde
{\varphi}\left(  \boldsymbol{i}\right)  \Omega\left(  p^{M}\left\Vert
x-\boldsymbol{i}\right\Vert _{p}\right)  =P_{N}\widetilde{\varphi}\left(
x\right)  $. Now
\[
\left\Vert \varphi-P_{N}\varphi\right\Vert _{\infty}\leq\left\Vert
\varphi-\widetilde{\varphi}\right\Vert _{\infty}+\left\Vert \widetilde
{\varphi}-P_{N}\varphi\right\Vert _{\infty}\leq\epsilon+\left\Vert
P_{N}\widetilde{\varphi}-P_{N}\varphi\right\Vert _{\infty}\leq2\epsilon\text{
\ }%
\]
for $N\geq\max\left\{  k,M\right\}  $, since $\left\Vert P_{N}\widetilde
{\varphi}-P_{N}\varphi\right\Vert _{\infty}=\sup_{\boldsymbol{i}}\left\vert
\widetilde{\varphi}\left(  \boldsymbol{i}\right)  -\varphi\left(
\boldsymbol{i}\right)  \right\vert \leq\left\Vert \widetilde{\varphi}%
-\varphi\right\Vert _{\infty}<\epsilon$ for $N\geq\max\left\{  k,M\right\}  $.
\end{proof}

We denote by $E_{N}$ , $N\geq1$, the embedding $X_{N}\rightarrow X_{\infty} $.
The following result is a consequence of the above observations. If $Z$, $Y$
are real Banach spaces, we denote by $\mathfrak{B}(Z,Y)$, the space of all
linear bounded operators from $Z$ into $Y$.

\begin{lemma}
[Condition A]\label{Condition A}With the above notation, the following
assertions hold:

\noindent(i) $X_{\infty}$, $X_{N}$ for $N\geq1$, are \ real Banach spaces, all
with the norm $\left\Vert \cdot\right\Vert _{\infty}$;

\noindent(ii) $P_{N}\in\mathfrak{B}\left(  X_{\infty},X_{N}\right)  $ and
$\left\Vert P_{N}\varphi\right\Vert _{\infty}\leq\left\Vert \varphi\right\Vert
_{\infty}$ for any $N\geq1$, $\varphi\in X_{\infty}$;

\noindent(iii) $E_{N}\in\mathfrak{B}\left(  X_{N},X_{\infty}\right)  $ and
$\left\Vert E_{N}\varphi\right\Vert _{\infty}=\left\Vert \varphi\right\Vert
_{\infty}$ for any $N\geq1$, $\varphi\in X_{N}$;

\noindent(iv) $P_{N}E_{N}\varphi=\varphi$ for $N\geq1$, $\varphi\in X_{N}$.
\end{lemma}

\subsection{\label{Sect_oper_A}The operators $A_{N}$, $A$}

Set $\mathbb{R}_{+}:=\{x\in\mathbb{R};x\geq0\}$. We fix a continuous function
$J:$\textbf{\ }$\mathbb{R}_{+}\rightarrow\mathbb{R}_{+}$, and take
$J(x)=J(||x||_{p})$ for $x\in%
%TCIMACRO{\U{211a} }%
%BeginExpansion
\mathbb{Q}
%EndExpansion
_{p}^{n}$, then $J(x)$ is a \textit{radial function} on $%
%TCIMACRO{\U{211a} }%
%BeginExpansion
\mathbb{Q}
%EndExpansion
_{p}^{n}$. In addition, we assume that $\ \int_{%
%TCIMACRO{\U{211a} }%
%BeginExpansion
\mathbb{Q}
%EndExpansion
_{p}^{n}}J(||x||_{p})d^{n}x=1$.$\ $

\begin{lemma}
\label{Lemma_D} The following assertions hold:

\noindent(i) set $J_{N}(||x||_{p}):=J(||x||_{p})\Omega\left(  p^{-N}\left\Vert
x\right\Vert _{p}\right)  $ for $N\geq1$. Then%
\[
J_{N}(||x||_{p})\ast P_{N}\varphi\left(  x\right)  =\Omega\left(
p^{-N}\left\Vert x\right\Vert _{p}\right)  \left\{  J(||x||_{p})\ast
P_{N}\varphi\left(  x\right)  \right\}
\]
for $\varphi\left(  x\right)  \in X_{\infty}$.

\noindent(ii) Define for $N\geq1$,%
\[%
\begin{array}
[c]{llll}%
A_{N}: & X_{N} & \rightarrow & X_{N}\\
& \phi\left(  x\right)  & \rightarrow & -\int\limits_{B_{N}^{n}}%
J_{N}(||x-y||_{p})\left\{  \phi\left(  y\right)  -\phi\left(  x\right)
\right\}  d^{n}y.
\end{array}
\]
Then $A_{N}$ is a well-defined linear bounded operator.
\end{lemma}

\begin{proof}
(i) We recall that $B_{N}^{n}\left(  0\right)  =B_{N}^{n}$. Notice that
\begin{align*}
I(x)  &  :=J_{N}(||x||_{p})\ast P_{N}\varphi\left(  x\right)  =\int\limits_{%
%TCIMACRO{\U{211a} }%
%BeginExpansion
\mathbb{Q}
%EndExpansion
_{p}^{n}}J\left(  \left\Vert x-y\right\Vert _{p}\right)  \Omega\left(
p^{-N}\left\Vert x-y\right\Vert _{p}\right)  P_{N}\varphi\left(  y\right)
d^{n}y\\
&  =\int\limits_{B_{N}^{n}\left(  0\right)  \cap B_{N}^{n}\left(  x\right)
}J\left(  \left\Vert x-y\right\Vert _{p}\right)  P_{N}\varphi\left(  y\right)
d^{n}y.
\end{align*}
The calculation of the above integral involves two cases: (1) $B_{N}%
^{n}\left(  0\right)  \cap B_{N}^{n}\left(  x\right)  \neq\varnothing$; (2)
$B_{N}^{n}\left(  0\right)  \cap B_{N}^{n}\left(  x\right)  =\varnothing$. In
the first case, since the radii of the balls are the same, $B_{N}^{n}\left(
0\right)  =B_{N}^{n}\left(  x\right)  $ and thus $x\in B_{N}^{n}\left(
0\right)  $ which implies that $\left\Vert x\right\Vert _{p}\leq p^{N}$, in
addition, supp$I(x)\subset B_{N}^{n}\left(  0\right)  $, and thus
\begin{align*}
I(x)  &  =\Omega\left(  p^{-N}\left\Vert x\right\Vert _{p}\right)
\int\limits_{B_{N}^{n}\left(  0\right)  }J\left(  \left\Vert x-y\right\Vert
_{p}\right)  P_{N}\varphi\left(  y\right)  d^{n}y\\
&  =\Omega\left(  p^{-N}\left\Vert x\right\Vert _{p}\right)  \int\limits_{%
%TCIMACRO{\U{211a} }%
%BeginExpansion
\mathbb{Q}
%EndExpansion
_{p}^{n}}J\left(  \left\Vert x-y\right\Vert _{p}\right)  P_{N}\varphi\left(
y\right)  d^{n}y.
\end{align*}
In the second case, $x\notin B_{N}^{n}\left(  0\right)  $ this implies that
$\left\Vert x\right\Vert _{p}>p^{-N}$, and $I(x)=0$.

(ii) Notice that
\begin{equation}
A_{N}\phi\left(  x\right)  =-\left\{  J_{N}(||x||_{p})\ast\phi\left(
x\right)  -j_{N}\phi\left(  x\right)  \right\}  ,\text{ with }j_{N}%
:=\int_{B_{N}^{n}}J(||y||_{p})d^{n}y, \label{formula_A_N}%
\end{equation}
$0\leq j_{N}\leq1$ and that%
\begin{equation}
\lim_{N\rightarrow\infty}j_{N}=1. \label{limit}%
\end{equation}
From (\ref{formula_A_N}) and part (i), it follows that the operator is
well-defined.\ Indeed, for $\phi\in X_{N}$, supp $\left(  J_{N}(||x||_{p}%
)\ast\phi\left(  x\right)  \right)  \subset B_{N}^{n}$ and the index of local
constancy of $J_{N}(||x||_{p})\ast\phi\left(  x\right)  $ equals the index of
local constancy of $\phi$. The continuity follows from the Young inequality:%
\[
\left\vert A_{N}\phi\left(  x\right)  \right\vert \leq\left\Vert
J_{N}(||x||_{p})\right\Vert _{L^{1}}\left\Vert \phi\left(  x\right)
\right\Vert _{\infty}+j_{N}\left\Vert \phi\left(  x\right)  \right\Vert
_{\infty}\leq2\left\Vert \phi\left(  x\right)  \right\Vert _{\infty}\text{.}%
\]

\end{proof}

Now, we define%
\begin{equation}%
\begin{array}
[c]{llll}%
A: & X_{\infty} & \rightarrow & X_{\infty}\\
& \varphi\left(  x\right)  & \rightarrow & A\varphi\left(  x\right)
=-\left\{  J\left(  \left\Vert x\right\Vert _{p}\right)  \ast\varphi\left(
x\right)  -\varphi\left(  x\right)  \right\}  .
\end{array}
\label{Operator_A}%
\end{equation}

\begin{remark}
Notice that $A\varphi\left(  x\right)  =-\int_{%
%TCIMACRO{\U{211a} }%
%BeginExpansion
\mathbb{Q}
%EndExpansion
_{p}^{n}}J\left(  \left\Vert x-y\right\Vert _{p}\right)  \left\{
\varphi\left(  y\right)  -\varphi\left(  x\right)  \right\}  d^{n}y$ since
$\int_{%
%TCIMACRO{\U{211a} }%
%BeginExpansion
\mathbb{Q}
%EndExpansion
_{p}^{n}}J\left(  \left\Vert x-y\right\Vert _{p}\right)  d^{n}y=1$.
\end{remark}

\begin{lemma}
\label{Lemma_E}The operator $A:X_{\infty}\rightarrow X_{\infty}$ is a linear
and bounded. In addition, the spectrum of $A$, $\sigma\left(  A\right)  $, is
contained in the interval $\left[  0,2\right]  $.
\end{lemma}

\begin{proof}
By the Young inequality, $A$ $\in\mathfrak{B}\left(  X_{\infty},L^{\infty
}\right)  $. Now, by construction $\mathcal{D}(%
%TCIMACRO{\U{211a} }%
%BeginExpansion
\mathbb{Q}
%EndExpansion
_{p}^{n})$ is dense in $X_{\infty}\left(
%TCIMACRO{\U{211a} }%
%BeginExpansion
\mathbb{Q}
%EndExpansion
_{p}^{n}\right)  $ with respect to $\left\Vert \cdot\right\Vert _{\infty}$,
then in order to show that $A$ is densely defined and continuous, from
$X_{\infty}$ into itself, it is sufficient to show that $A\varphi\in
X_{\infty}$ for $\varphi\in\mathcal{D}_{N}^{-N}\subset\mathcal{D}(%
%TCIMACRO{\U{211a} }%
%BeginExpansion
\mathbb{Q}
%EndExpansion
_{p}^{n})$. Thus, we have to show that $A\varphi\in X_{\infty}$ for
$\varphi\in X_{N}$. By Lemma \ref{Lemma_D}, $\Omega\left(  p^{-N}\left\Vert
x\right\Vert _{p}\right)  A\varphi=A_{N}\varphi\in X_{N}$, now we show that
$A_{N}\varphi$ $\underrightarrow{\left\Vert \cdot\right\Vert _{\infty}}$
$A\varphi$ for $\varphi\in X_{N}$. In order to achieve this, by Lemma
\ref{Lemma_0}, it is sufficient to show that $J_{N}\left(  \left\Vert
x\right\Vert _{p}\right)  \ast\varphi\left(  x\right)  $ $\underrightarrow
{\left\Vert \cdot\right\Vert _{\infty}}$ $J\left(  \left\Vert x\right\Vert
_{p}\right)  \ast\varphi\left(  x\right)  $ for $\varphi\in X_{N}$. Indeed,%
\[
\left\Vert \left\{  J\left(  \left\Vert x\right\Vert _{p}\right)
-J_{N}\left(  \left\Vert x\right\Vert _{p}\right)  \right\}  \ast
\varphi\left(  x\right)  \right\Vert _{\infty}\leq\left\Vert \varphi
\right\Vert _{\infty}\int\limits_{%
%TCIMACRO{\U{211a} }%
%BeginExpansion
\mathbb{Q}
%EndExpansion
_{p}^{n}}\left\vert J\left(  \left\Vert y\right\Vert _{p}\right)
-J_{N}\left(  \left\Vert y\right\Vert _{p}\right)  \right\vert d^{n}y.
\]
This last integral tends to zero as $N$ tends to infinity by the Dominated
Convergence Theorem, recall that $\int_{%
%TCIMACRO{\U{211a} }%
%BeginExpansion
\mathbb{Q}
%EndExpansion
_{p}^{n}}J_{N}\left(  \left\Vert y\right\Vert _{p}\right)  d^{n}y\leq\int_{%
%TCIMACRO{\U{211a} }%
%BeginExpansion
\mathbb{Q}
%EndExpansion
_{p}^{n}}J\left(  \left\Vert y\right\Vert _{p}\right)  d^{n}y=1$.

The comment about the spectrum of $A$ \ follows from the \ following
observation: the equation $A\varphi=\lambda\varphi$ is equivalent to
$J\ast\varphi=\left(  1-\lambda\right)  \varphi$ and since $\left\Vert
J\ast\cdot\right\Vert \leq1$, we have $0\leq\lambda\leq2$.
\end{proof}

\section{\label{sect_matrix_repre}The Matrix Representation of operators
$A_{N}$ and Markov Chains}

By using the basis $\left\{  \Omega\left(  p^{N}\left\Vert x-\boldsymbol{i}%
\right\Vert _{p}\right)  \right\}  _{\boldsymbol{i}\in G_{N}^{n}}$, we
identify $X_{N}$ with $\left(  \mathbb{R}^{\#G_{N}^{n}},\left\Vert
\cdot\right\Vert _{\mathbb{R}}\right)  $, thus operator $A_{_{N}}$ is given by
a matrix. This matrix is computed by means of the following two lemmas.

\begin{lemma}
\label{Lemma_B}Set $\mathfrak{a}\left(  x,\boldsymbol{i}\right)
:=J_{N}\left(  \left\Vert x\right\Vert _{p}\right)  \ast\Omega\left(
p^{N}\left\Vert x-\boldsymbol{i}\right\Vert _{p}\right)  $ for $x\in B_{N}%
^{n}$, $\boldsymbol{i}\in G_{N}^{n}$. Let $\widetilde{x}$ denote the image of
$x$ under the canonical map $B_{N}^{n}\rightarrow G_{N}^{n}$. Then
\[
\mathfrak{a}\left(  x,\boldsymbol{i}\right)  =\mathfrak{a}\left(
\widetilde{x},\boldsymbol{i}\right)  =\left\{
\begin{array}
[c]{ccc}%
p^{-Nn}J\left(  p^{-ord(\widetilde{x}-\boldsymbol{i})}\right)  & \text{if} &
ord(\widetilde{x}-\boldsymbol{i})\neq+\infty\\
&  & \\
\int\limits_{\left(  p^{N}\mathbb{Z}_{p}\right)  ^{n}}J\left(  \left\Vert
y\right\Vert _{p}\right)  d^{n}y & \text{if} & ord(\widetilde{x}%
-\boldsymbol{i})=+\infty.
\end{array}
\right.
\]

\end{lemma}

\begin{proof}
We first notice that since $\widetilde{x}-\boldsymbol{i}+\left(
p^{N}\mathbb{Z}_{p}\right)  ^{n}\subset B_{N}^{n}$, it \ verifies that
\begin{align}
\mathfrak{a}\left(  x,\boldsymbol{i}\right)   &  =\int
\limits_{x-\boldsymbol{i}+\left(  p^{N}\mathbb{Z}_{p}\right)  ^{n}}%
J_{N}\left(  \left\Vert y\right\Vert _{p}\right)  d^{n}y=\int
\limits_{\widetilde{x}-\boldsymbol{i}+\left(  p^{N}\mathbb{Z}_{p}\right)
\ ^{n}}J_{N}\left(  \left\Vert y\right\Vert _{p}\right)  d^{n}y\nonumber\\
&  =\int\limits_{\widetilde{x}-\boldsymbol{i}+\left(  p^{N}\mathbb{Z}%
_{p}\right)  ^{n}}J\left(  \left\Vert y\right\Vert _{p}\right)  d^{n}%
y=J\left(  \left\Vert \widetilde{x}\right\Vert _{p}\right)  \ast\Omega\left(
p^{N}\left\Vert \widetilde{x}-\boldsymbol{i}\right\Vert _{p}\right)
\nonumber\\
&  =\int\limits_{\left(  p^{N}\mathbb{Z}_{p}\right)  ^{n}}J\left(  \left\Vert
\widetilde{x}-\boldsymbol{i}-y\right\Vert _{p}\right)  d^{n}y.
\label{Formula_H}%
\end{align}
On the other hand, since $ord(G_{N}^{n})=\left\{  -N,-N+1,\ldots
,0,1,\ldots,N-1,+\infty\right\}  $ and $ord(y)\geq N$ for $y\in\left(
p^{N}\mathbb{Z}_{p}\right)  ^{n}$, then $\left\Vert \widetilde{x}%
-\boldsymbol{i}-y\right\Vert _{p}=\left\Vert \widetilde{x}-\boldsymbol{i}%
\right\Vert _{p}$ if and only if $ord\left(  \widetilde{x}-\boldsymbol{i}%
\right)  \neq+\infty$. The announced formula now follows from (\ref{Formula_H}).
\end{proof}

\begin{remark}
Notice that $\mathfrak{a}\left(  \widetilde{x},\boldsymbol{i}\right)
=\mathfrak{a}\left(  \boldsymbol{i},\widetilde{x}\right)  =\mathfrak{a}\left(
\left\Vert \widetilde{x}-\boldsymbol{i}\right\Vert _{p}\right)  $, where
$\mathfrak{a}\left(  \left\Vert \widetilde{x}-\boldsymbol{i}\right\Vert
_{p}\right)  $ means that there exists a function $g:\mathbb{R}_{+}%
\mathbb{\rightarrow R}$ such that $\mathfrak{a}\left(  \boldsymbol{i}%
,\widetilde{x}\right)  =g\left(  \left\Vert \widetilde{x}-\boldsymbol{i}%
\right\Vert _{p}\right)  $, i.e. $\mathfrak{a}\left(  \boldsymbol{i}%
,\widetilde{x}\right)  $ is a radial function \ of $\widetilde{x}%
-\boldsymbol{i}$.
\end{remark}

\begin{lemma}
\label{Lemma_C}The matrix for operator $A_{N}$ acting on $X_{N}$ is
$A^{\left(  N\right)  }=\left[  A_{\boldsymbol{ki}}^{\left(  N\right)
}\right]  _{\boldsymbol{k},\boldsymbol{i}\in G_{N}^{n}}=\left[  j_{N}%
\delta_{\boldsymbol{ki}}-\mathfrak{a}_{\boldsymbol{ki}}\right]
_{\boldsymbol{k},\boldsymbol{i}\in G_{N}^{n}}$, where $\mathfrak{a}%
_{\boldsymbol{ki}}:=\mathfrak{a}(\boldsymbol{k},\boldsymbol{i})$ and
$\delta_{\boldsymbol{ki}}$ denotes the Kronecker delta.
\end{lemma}

\begin{proof}
Notice that $\left\{  \Omega\left(  p^{N}\left\Vert x-\boldsymbol{i}%
\right\Vert _{p}\right)  \right\}  _{\boldsymbol{i}\in G_{N}^{n}}$ is an
orthogonal basis of $X_{N}$ under the usual inner product of real-valued
functions, which is denoted as $\left\langle \cdot,\cdot\right\rangle $,
since
\[
\Omega\left(  p^{N}\left\Vert x-\boldsymbol{i}\right\Vert _{p}\right)
\Omega\left(  p^{N}\left\Vert x-\boldsymbol{j}\right\Vert _{p}\right)
\equiv0\text{ if }\boldsymbol{i\neq j}.
\]
Take $\varphi\left(  x\right)  =\sum_{\boldsymbol{i}\in G_{N}^{n}}%
\varphi\left(  \boldsymbol{i}\right)  \Omega\left(  p^{N}\left\Vert
x-\boldsymbol{i}\right\Vert _{p}\right)  $, then
\[
J_{N}\left(  \left\Vert x\right\Vert _{p}\right)  \ast\Omega\left(
p^{N}\left\Vert x-\boldsymbol{i}\right\Vert _{p}\right)  =\sum
\limits_{\boldsymbol{k}\in G_{N}^{n}}d_{\boldsymbol{ki}}\Omega\left(
p^{N}\left\Vert x-\boldsymbol{k}\right\Vert _{p}\right)  ,
\]
with $d_{\boldsymbol{ki}}\in\mathbb{R}$, cf. Lemma \ref{Lemma_D}, and%
\[
A_{N}\varphi\left(  x\right)  =\sum\limits_{\boldsymbol{k}\in G_{N}^{n}%
}\left\{  \sum\limits_{\boldsymbol{i}\in G_{N}^{n}}\left(  j_{N}%
\delta_{\boldsymbol{ki}}-d_{\boldsymbol{ki}}\right)  \varphi\left(
\boldsymbol{i}\right)  \right\}  \Omega\left(  p^{N}\left\Vert
x-\boldsymbol{k}\right\Vert _{p}\right)  .
\]
Hence $\left[  j_{N}\delta_{\boldsymbol{ki}}-d_{\boldsymbol{ki}}\right]
_{\boldsymbol{k},\boldsymbol{i}\in G_{N}^{n}}$ is the matrix representation
for operator $A_{N}$ acting on $X_{N}$. To compute the coefficients
$d_{\boldsymbol{ki}}$ we proceed as follows. By using that%
\[
\Omega\left(  p^{N}\left\Vert y-\boldsymbol{j}\right\Vert _{p}\right)
\ast\Omega\left(  p^{N}\left\Vert y-\boldsymbol{l}\right\Vert _{p}\right)
=p^{-Nn}\Omega\left(  p^{N}\left\Vert y-\left(  \boldsymbol{j}+\boldsymbol{l}%
\right)  \right\Vert _{p}\right)  ,
\]
we have%
\begin{gather}
d_{\boldsymbol{ki}}=P^{Nn}\left\langle J_{N}\left(  \left\Vert x\right\Vert
_{p}\right)  \ast\Omega\left(  p^{N}\left\Vert x-\boldsymbol{i}\right\Vert
_{p}\right)  ,\Omega\left(  p^{N}\left\Vert x-\boldsymbol{k}\right\Vert
_{p}\right)  \right\rangle \nonumber\\
=p^{Nn}\int\limits_{%
%TCIMACRO{\U{211a} }%
%BeginExpansion
\mathbb{Q}
%EndExpansion
_{p}^{n}}\int\limits_{%
%TCIMACRO{\U{211a} }%
%BeginExpansion
\mathbb{Q}
%EndExpansion
_{p}^{n}}J_{N}\left(  \left\Vert y\right\Vert _{p}\right)  \Omega\left(
p^{N}\left\Vert x-\boldsymbol{i}-y\right\Vert _{p}\right)  \Omega\left(
p^{N}\left\Vert x-\boldsymbol{k}\right\Vert _{p}\right)  d^{n}yd^{n}%
x\nonumber\\
=p^{Nn}\int\limits_{%
%TCIMACRO{\U{211a} }%
%BeginExpansion
\mathbb{Q}
%EndExpansion
_{p}^{n}}J_{N}\left(  \left\Vert y\right\Vert _{p}\right)  \Omega\left(
p^{N}\left\Vert y+\boldsymbol{i}\right\Vert _{p}\right)  \ast\Omega\left(
p^{N}\left\Vert y-\boldsymbol{k}\right\Vert _{p}\right)  d^{n}y\nonumber\\
=\int\limits_{%
%TCIMACRO{\U{211a} }%
%BeginExpansion
\mathbb{Q}
%EndExpansion
_{p}^{n}}J_{N}\left(  \left\Vert y\right\Vert _{p}\right)  \Omega\left(
p^{N}\left\Vert y-\left(  \boldsymbol{k}-\boldsymbol{i}\right)  \right\Vert
_{p}\right)  d^{n}y=\int\limits_{\boldsymbol{k}-\boldsymbol{i}+\left(
p^{N}\mathbb{Z}_{p}\right)  ^{n}}J\left(  \left\Vert y\right\Vert _{p}\right)
d^{n}y\nonumber\\
=\mathfrak{a}(\boldsymbol{k},\boldsymbol{i}), \label{Calculo}%
\end{gather}
cf. Lemma \ref{Lemma_B} .
\end{proof}

\begin{lemma}
\label{Lemma_D1}$-A^{\left(  N\right)  }$ is a $Q$-matrix, i.e.
$-A_{\boldsymbol{ij}}^{\left(  N\right)  }\geq0$ for $\boldsymbol{i}%
\neq\boldsymbol{j}$ with $\boldsymbol{i}$, $\boldsymbol{j}\in G_{N}^{n}$, and
$A_{\boldsymbol{ii}}^{\left(  N\right)  }=-\sum_{\boldsymbol{j}\neq
\boldsymbol{i}}A_{\boldsymbol{ij}}^{\left(  N\right)  }$.
\end{lemma}

\begin{proof}
We first notice that
\[
j_{N}-\mathfrak{a}_{\boldsymbol{ii}}=\int\limits_{B_{N}^{n}}J\left(
\left\Vert y\right\Vert _{p}\right)  d^{n}y-\int\limits_{B_{-N}^{n}}J\left(
\left\Vert y\right\Vert _{p}\right)  d^{n}y\geq0,
\]
cf. Lemmas \ref{Lemma_C}, \ref{Lemma_B}. Now, $A_{\boldsymbol{ij}}^{\left(
N\right)  }=-\mathfrak{a}_{\boldsymbol{ij}}$ for $\boldsymbol{j}%
\neq\boldsymbol{i}$ with $\mathfrak{a}_{\boldsymbol{ij}}\geq0$, and by using
that $G_{N}^{n}$ is an additive group, and that $B_{N}^{n}=%
%TCIMACRO{\tcoprod \nolimits_{\boldsymbol{k}\in G_{N}^{n}}}%
%BeginExpansion
{\textstyle\coprod\nolimits_{\boldsymbol{k}\in G_{N}^{n}}}
%EndExpansion
B_{-N}^{n}\left(  \boldsymbol{k}\right)  $, we have%
\begin{align*}
\sum\limits_{\boldsymbol{j}\neq\boldsymbol{i}}\mathfrak{a}_{\boldsymbol{ij}}
&  =\sum\limits_{\boldsymbol{j}\neq\boldsymbol{i}}\text{ }\int\limits_{B_{-N}%
^{n}\left(  \boldsymbol{j}-\boldsymbol{i}\right)  }J\left(  \left\Vert
y\right\Vert _{p}\right)  d^{n}y=\sum\limits_{\boldsymbol{k}\neq
\boldsymbol{0}}\int\limits_{B_{-N}^{n}\left(  \boldsymbol{k}\right)  }J\left(
\left\Vert y\right\Vert _{p}\right)  d^{n}y\\
&  =\int\limits_{B_{N}^{n}}J\left(  \left\Vert y\right\Vert _{p}\right)
d^{n}y-\int\limits_{B_{-N}^{n}}J\left(  \left\Vert y\right\Vert _{p}\right)
d^{n}y=j_{N}-\mathfrak{a}_{\boldsymbol{ii}}\text{,}%
\end{align*}
i.e. $A_{\boldsymbol{ii}}^{\left(  N\right)  }=-\sum_{\boldsymbol{j}%
\neq\boldsymbol{i}}A_{\boldsymbol{ij}}^{\left(  N\right)  }$.
\end{proof}

A real matrix $A$ is called \textit{nonnegative} if each of its entries is
greater than or equal to zero, in this case, we use the notation
$A\geq\boldsymbol{0}$. Similarly, we say that a real matrix is
\textit{nonpositive} if each of its entries is less than or equal to zero, in
this case, we use the notation $A\leq\boldsymbol{0}$. We denote by
$\mathbb{E}$ the identity matrix and by $\boldsymbol{1}$ the unit vector,
which is \ the vector having all its entries equal to one.

\begin{theorem}
\label{Theorem2B}(i) Set $P^{\left(  N\right)  }\left(  t\right)
:=e^{-tA^{\left(  N\right)  }}$, $t\geq0$. Then $P^{\left(  N\right)  }\left(
t\right)  $ is a semigroup of nonnegative matrices, with $P^{\left(  N\right)
}\left(  0\right)  =\mathbb{E}$, satisfying
\[
\frac{\partial P^{\left(  N\right)  }\left(  t\right)  }{\partial
t}+A^{\left(  N\right)  }P^{\left(  N\right)  }\left(  t\right)  =0,
\]
and $P^{\left(  N\right)  }\left(  t\right)  \boldsymbol{1}=\boldsymbol{1}$
for $t\geq0$.

\noindent(ii) The function $P^{\left(  N\right)  }\left(  t-s\right)  $,
$t\geq s\geq0$, is the transition function of a homogeneous Markov chain with
state space $G_{N}^{n}$. Furthermore, this stochastic process has
right-continuous piece-wise-constant paths.
\end{theorem}

\begin{proof}
The result follows from Lemma \ref{Lemma_D1} by using well-known results about
Markov chains, see e.g. \cite[Theorem 2.5]{Yin-Zhang}.
\end{proof}

\section{\label{Sect_G_Landau}Non-Archimedean Helmholtz Free-Energy
Functionals}

We define for $\varphi\in X_{N}$, and $\lambda>0$,%
\begin{equation}
E_{N}\left(  \varphi\right)  =\frac{1}{4}\int\limits_{B_{N}^{n}}%
\int\limits_{B_{N}^{n}}J_{N}\left(  \left\Vert x-y\right\Vert _{p}\right)
\left\{  \varphi\left(  x\right)  -\varphi\left(  y\right)  \right\}
^{2}d^{n}xd^{n}y+\lambda\int\limits_{B_{N}^{n}}W\left(  \varphi\left(
x\right)  \right)  d^{n}x, \label{E_N}%
\end{equation}
where $J_{N}\left(  \left\Vert x\right\Vert _{p}\right)  $ is as before,
$\varphi$ is a scalar density function defined on $B_{N}^{n}$ that takes
values in $\left[  -1,1\right]  $, $W:\mathbb{R}\rightarrow\mathbb{R}$, with
derivative $f\in C^{2}\left(  \mathbb{R}\right)  $, is a double-well potential
having (not necessarily equal) minima at $\pm1$. The functional $E_{N}\left(
\varphi\right)  $ is a non-Archimedean version\ of a non-local Helmholtz
free-energy functional. The function $\varphi$, the order parameter,
represents the macroscopic density profile of a system which has two
equilibrium pure phases described by the profiles $\varphi\equiv1$ and
$\varphi\equiv-1$, and $-1<\varphi<1$ represents the `interface'. The function
$J_{N}$ is a positive, possibly anisotropic, interaction potential which
vanishes at infinity. If $\varphi$ is an energy minimizing configuration, the
second term in $E_{N}$ forces the minimizer $\varphi$ to take values close the
pure states, while the first term in $E_{N}$ represents an interaction energy
which penalizes the spatial inhomogenety of $\varphi$.

In the classical Archimedean setting (i.e. $\mathbb{R}^{n}$), the $L^{2}%
$-gradient of functionals of type (\ref{E_N}) lead to the non-local versions
of Allen-Cahn equations, see \cite{Alberti-B-1}-\cite{Alberti-B-2},
\cite{BatesChmdj} The next result shows that a similar situation \ happens in
the non-Archimedean setting.

\begin{lemma}
\label{Lemma_F}(i) By identifying $\varphi\left(  x\right)  $ with the vector
$\left[  \varphi\left(  \boldsymbol{i}\right)  \right]  _{\boldsymbol{i}\in
G_{N}^{n}}$, i.e. by identifying $X_{N}$ with $\mathbb{R}^{\#G_{N}^{n}}$, we
have
\begin{align*}
E_{N}\left(  \left[  \varphi\left(  \boldsymbol{i}\right)  \right]
_{\boldsymbol{i}\in G_{N}^{n}}\right)   &  =\frac{j_{N}p^{-Nn}}{2}%
\sum\limits_{\boldsymbol{i}\in G_{N}^{n}}\varphi^{2}\left(  \boldsymbol{i}%
\right)  -\frac{p^{-Nn}}{2}\sum\limits_{\boldsymbol{i},\boldsymbol{j}\in
G_{N}^{n}}\mathfrak{a}_{\boldsymbol{ij}}\varphi\left(  \boldsymbol{i}\right)
\varphi\left(  \boldsymbol{j}\right) \\
&  +\lambda p^{-Nn}\sum\limits_{\boldsymbol{i}\in G_{N}^{n}}W\left(
\varphi\left(  \boldsymbol{i}\right)  \right)  ,
\end{align*}
where $\left[  \mathfrak{a}_{\boldsymbol{ij}}\right]  _{\boldsymbol{i,j}\in
G_{N}^{n}}$ is the matrix defined in Lemma \ref{Lemma_B}.

\noindent(ii) We assume that $\varphi$ depends on $\boldsymbol{i}\in G_{N}%
^{n}$ and $t\geq0$. The gradient flow in the Euclidean space $\mathbb{R}%
^{\#G_{N}^{n}}$ of the functional $E_{N}:\mathbb{R}^{\#G_{N}^{n}}%
\rightarrow\mathbb{R}$ is the evolution in $\mathbb{R}^{\#G_{N}^{n}}$ given
by
\begin{align}
\frac{\partial}{\partial t}\left[  \varphi\left(  \boldsymbol{i},t\right)
\right]  _{\boldsymbol{i}\in G_{N}^{n}}  &  =-\nabla E_{N}\left(  \left[
\varphi\left(  \boldsymbol{i},t\right)  \right]  _{\boldsymbol{i}\in G_{N}%
^{n}}\right) \label{Eq_gradient}\\
&  =-p^{-Nn}A^{\left(  N\right)  }\left[  \varphi\left(  \boldsymbol{i}%
,t\right)  \right]  _{\boldsymbol{i}\in G_{N}^{n}}-\lambda p^{-Nn}\left[
f\left(  \varphi\left(  \boldsymbol{i},t\right)  \right)  \right]
_{\boldsymbol{i}\in G_{N}^{n}},\nonumber
\end{align}
where $A^{\left(  N\right)  }$ is the matrix defined in Lemma \ref{Lemma_C}.
\end{lemma}

\begin{remark}
Notice that in $X_{N}$, (\ref{Eq_gradient}) can be written as
\begin{equation}
\frac{\partial}{\partial t}\varphi\left(  x,t\right)  =-A_{N}\varphi\left(
x,t\right)  -\lambda f\left(  \varphi\left(  x,t\right)  \right)  .
\label{Eq_gradient_2}%
\end{equation}

\end{remark}

\begin{proof}
(i) By using that $\varphi^{2}\left(  x\right)  =\sum_{\boldsymbol{i}\in
G_{N}^{n}}\varphi^{2}\left(  \boldsymbol{i}\right)  \Omega\left(
p^{N}\left\Vert x-\boldsymbol{i}\right\Vert _{p}\right)  $, $\varphi\left(
x\right)  \varphi\left(  y\right)  =\sum_{\boldsymbol{i},\boldsymbol{j}\in
G_{N}^{n}}\varphi\left(  \boldsymbol{i}\right)  \varphi\left(  \boldsymbol{j}%
\right)  \Omega\left(  p^{N}\left\Vert x-\boldsymbol{i}\right\Vert
_{p}\right)  \Omega\left(  p^{N}\left\Vert y-\boldsymbol{j}\right\Vert
_{p}\right)  $, we have%
\begin{multline*}
E_{N}\left(  \left[  \varphi_{\boldsymbol{i}}\right]  _{\boldsymbol{i}\in
G_{N}^{n}}\right)  =\frac{1}{2}\sum_{\boldsymbol{i}\in G_{N}^{n}}\varphi
^{2}\left(  \boldsymbol{i}\right)  \int\limits_{B_{N}^{n}}\int\limits_{B_{N}%
^{n}}J\left(  \left\Vert x-y\right\Vert _{p}\right)  \Omega\left(
p^{N}\left\Vert x-\boldsymbol{i}\right\Vert _{p}\right)  d^{n}xd^{n}y\\
-\frac{1}{2}\sum_{\boldsymbol{i},\boldsymbol{j}\in G_{N}^{n}}\varphi\left(
\boldsymbol{i}\right)  \varphi\left(  \boldsymbol{j}\right)  \int
\limits_{B_{N}^{n}}\int\limits_{B_{N}^{n}}J\left(  \left\Vert x-y\right\Vert
_{p}\right)  \Omega\left(  p^{N}\left\Vert x-\boldsymbol{i}\right\Vert
_{p}\right)  \Omega\left(  p^{N}\left\Vert y-\boldsymbol{j}\right\Vert
_{p}\right)  d^{n}xd^{n}y\\
+\lambda\sum_{\boldsymbol{i}\in G_{N}^{n}}W\left(  \varphi\left(
\boldsymbol{i}\right)  \right)  \int\limits_{B_{N}^{n}}\Omega\left(
p^{N}\left\Vert x-\boldsymbol{i}\right\Vert _{p}\right)  d^{n}x.
\end{multline*}
The announced formula follows from the following observations:%
\begin{align*}
&  \int\limits_{B_{N}^{n}}\Omega\left(  p^{N}\left\Vert x-\boldsymbol{i}%
\right\Vert _{p}\right)  \left\{  \int\limits_{B_{N}^{n}}J\left(  \left\Vert
x-y\right\Vert _{p}\right)  d^{n}y\right\}  d^{n}x\\
&  =\int\limits_{B_{N}^{n}}\Omega\left(  p^{N}\left\Vert x-\boldsymbol{i}%
\right\Vert _{p}\right)  \left\{  \int\limits_{B_{N}^{n}}J\left(  \left\Vert
z\right\Vert _{p}\right)  d^{n}z\right\}  d^{n}x=p^{-Nn}j_{N},
\end{align*}
and
\begin{align*}
&  \int\limits_{B_{N}^{n}}\int\limits_{B_{N}^{n}}J\left(  \left\Vert
x-y\right\Vert _{p}\right)  \Omega\left(  p^{N}\left\Vert x-\boldsymbol{i}%
\right\Vert _{p}\right)  \Omega\left(  p^{N}\left\Vert y-\boldsymbol{j}%
\right\Vert _{p}\right)  d^{n}xd^{n}y\\
&  =\left\langle J\left(  \left\Vert x\right\Vert _{p}\right)  \ast
\Omega\left(  p^{N}\left\Vert x-\boldsymbol{i}\right\Vert _{p}\right)
,\Omega\left(  p^{N}\left\Vert x-\boldsymbol{j}\right\Vert _{p}\right)
\right\rangle =p^{-Nn}\mathfrak{a}_{\boldsymbol{ij}},
\end{align*}
see (\ref{Calculo}).

(ii) By using the first part, with $\left\langle \cdot,\cdot\right\rangle
_{\mathbb{R}}$ denoting the inner product in $\mathbb{R}^{\#G_{N}^{n}}$, the
directional derivative (i.e. functional derivative) of $E_{N}$ is given by%
\begin{multline*}
\lim_{\epsilon\rightarrow0}\frac{E_{N}\left(  \left[  \varphi\left(
\boldsymbol{i}\right)  +\epsilon\theta\left(  \boldsymbol{i}\right)  \right]
_{\boldsymbol{i}\in G_{N}^{n}}\right)  -E_{N}\left(  \left[  \varphi\left(
\boldsymbol{i}\right)  \right]  _{\boldsymbol{i}\in G_{N}^{n}}\right)
}{\epsilon}\\
=j_{N}p^{-Nn}\sum\limits_{\boldsymbol{i}}\varphi\left(  \boldsymbol{i}\right)
\theta\left(  \boldsymbol{i}\right)  -p^{-Nn}\sum\limits_{\boldsymbol{i}%
,\boldsymbol{j}}\mathfrak{a}_{\boldsymbol{ij}}\varphi\left(  \boldsymbol{j}%
\right)  \theta\left(  \boldsymbol{i}\right)  +\lambda p^{-Nn}\sum
\limits_{\boldsymbol{i}}f\left(  \varphi\left(  \boldsymbol{i}\right)
\right)  \theta\left(  \boldsymbol{i}\right) \\
=p^{-Nn}\left\langle A^{(N)}\left[  \varphi\left(  \boldsymbol{i}\right)
\right]  _{\boldsymbol{i}\in G_{N}^{n}}+\lambda\left[  f\left(  \varphi\left(
\boldsymbol{i}\right)  \right)  \right]  _{\boldsymbol{i}\in G_{N}^{n}%
},\left[  \theta\left(  \boldsymbol{i}\right)  \right]  _{\boldsymbol{i}\in
G_{N}^{n}}\right\rangle _{\mathbb{R}},
\end{multline*}
i.e. $\nabla E_{N}\left(  \left[  \varphi\left(  \boldsymbol{i}\right)
\right]  _{\boldsymbol{i}\in G_{N}^{n}}\right)  =p^{-Nn}\left(  A^{(N)}\left[
\varphi\left(  \boldsymbol{i}\right)  \right]  _{\boldsymbol{i}\in G_{N}^{n}%
}+\lambda\left[  f\left(  \varphi\left(  \boldsymbol{i}\right)  \right)
\right]  _{\boldsymbol{i}\in G_{N}^{n}}\right)  $ where $\nabla g$ denotes the
standard gradient vector in $\mathbb{R}^{\#G_{N}^{n}}$. On the other hand, in
the space $X_{N}$, we have $\nabla E_{N}\left(  \varphi\right)  =A_{N}%
\varphi+\lambda f\left(  \varphi\right)  $ since
\begin{align*}
&  p^{-Nn}\left\langle A^{(N)}\left[  \varphi\left(  \boldsymbol{i}\right)
\right]  _{\boldsymbol{i}\in G_{N}^{n}}+\lambda\left[  f\left(  \varphi\left(
\boldsymbol{i}\right)  \right)  \right]  _{\boldsymbol{i}\in G_{N}^{n}%
},\left[  \theta\left(  \boldsymbol{i}\right)  \right]  _{\boldsymbol{i}\in
G_{N}^{n}}\right\rangle _{\mathbb{R}}\\
&  =\int_{\mathbb{Q}_{p}^{n}}\left\{  A_{N}\varphi\left(  x\right)  +\lambda
f\left(  \varphi\left(  x\right)  \right)  \right\}  \theta\left(  x\right)
d^{n}x
\end{align*}
in $X_{N}$.
\end{proof}

Consider $\left(  G_{N}^{n},\left\Vert \cdot\right\Vert _{p}\right)  $ as a
finite ultrametric space. Then (\ref{Eq_gradient}) is
reaction-ultradi\-ffusion equation in $\left(  G_{N}^{n},\left\Vert
\cdot\right\Vert _{p}\right)  $, which is the $L^{2}$-gradient of an energy
functional defined on $\left(  G_{N}^{n},\left\Vert \cdot\right\Vert
_{p}\right)  $. These equations are generalizations of the ultradiffusion
equations studied in \cite{O-S} and \cite{Bachas-Hu}. In this article we
initiate the study of these equations and their `limits'\ as $N$ tends to
infinity. The limit of \ some ultradiffusion equations of type
(\ref{Eq_gradient}) with $f\equiv0$ was considered by Avetisov et al. in
\cite{Av-8}, when the matrix $A^{\left(  N\right)  }$ comes from a Parisi
matrix. More precisely, in \cite{Av-8} \ was\ established, by using a physical
argument, that the `limit'\ of an equation of type (\ref{Eq_gradient}) as $N$
tends to infinity is
\begin{equation}
\frac{\partial}{\partial t}\varphi\left(  x,t\right)  =-A\varphi\left(
x,t\right)  -\lambda f\left(  \varphi\left(  x,t\right)  \right)  \text{,
}x\in\mathbb{Q}_{p}^{n}\text{, }t\geq0. \label{Eq_gradient_3}%
\end{equation}
In this article we show, from a mathematical perspective, that the solutions
of the Cauchy problem attached to the equation (\ref{Eq_gradient_2}) converge
to the solutions of the Cauchy problem attached to the equation
(\ref{Eq_gradient_3}), see Theorem \ref{Theorem4}, in the case that $f$ $\in
C^{2}$ with three zeros at $-1$, $0$, $1$. The equation (\ref{Eq_gradient_3})
is formally the $L^{2}$-gradient of the following energy functional:
\[
E\left(  \varphi\right)  =\frac{1}{4}\int\limits_{\mathbb{Q}_{p}^{n}}%
\int\limits_{\mathbb{Q}_{p}^{n}}J\left(  \left\Vert x-y\right\Vert
_{p}\right)  \left\{  \varphi\left(  x\right)  -\varphi\left(  y\right)
\right\}  ^{2}d^{n}xd^{n}y+\lambda\int\limits_{\mathbb{Q}_{p}^{n}}W\left(
\varphi\left(  x\right)  \right)  d^{n}x,
\]
where $\varphi$ is a scalar density function defined on $\mathbb{Q}_{p}^{n}$
that takes values in $\left[  -1,1\right]  $, $W$ is a double-well potential
having minima at $\pm1$ as before. At the moment, we do not know if $E\left(
\varphi\right)  $ can be well approximated by $E_{N}\left(  P_{N}%
\varphi\right)  $ for $\varphi\in X_{\infty}$. On the other hand, if supp
$J\subset K$, open and compact, with $K\subset B_{N_{0}}^{n}$, and $W^{\prime
}=f$ is continuous, then for any $\varphi\in X_{\infty}\left(  K\right)  $,
and $N\geq N_{0}$, the functional
\begin{align*}
E_{N}\left(  P_{N}\varphi\right)   &  =\frac{1}{4}\int\limits_{K}%
\int\limits_{K}J\left(  \left\Vert x-y\right\Vert _{p}\right)  \left\{
P_{N}\varphi\left(  x\right)  -P_{N}\varphi\left(  y\right)  \right\}
^{2}d^{n}xd^{n}y\\
&  +\lambda\int\limits_{K}W\left(  P_{N}\varphi\left(  x\right)  \right)
d^{n}x
\end{align*}
tends to
\[
\frac{1}{4}\int\limits_{K}\int\limits_{K}J\left(  \left\Vert x-y\right\Vert
_{p}\right)  \left\{  \varphi\left(  x\right)  -\varphi\left(  y\right)
\right\}  ^{2}d^{n}xd^{n}y+\lambda\int\limits_{K}W\left(  \varphi\left(
x\right)  \right)  d^{n}x
\]
as $N$ tends to infinity. The verification of this assertion follows directly
from the Dominated Convergence Theorem and the fact that $P_{N}\varphi$
$\underrightarrow{\left\Vert .\right\Vert _{\infty}}$ $\varphi$.

\section{\label{Sect_Staionary_Solutions}Stationary Solutions}

We take $J(x)=J(||x||_{p})$ for $x\in%
%TCIMACRO{\U{211a} }%
%BeginExpansion
\mathbb{Q}
%EndExpansion
_{p}^{n}$ as in Section \ref{Sect_oper_A}. We fix a function $f:\mathbb{R}%
\rightarrow\mathbb{R}$ having the following properties:%
\begin{equation}
f\in C^{2}\left(  \mathbb{R}\right)  \text{;} \tag{H1}%
\end{equation}%
\begin{equation}
f\text{ has exactly three zeros at }-1\text{, }0\text{, }1\text{;} \tag{H2}%
\end{equation}%
\begin{equation}
f^{\prime}\left(  -1\right)  >0\text{, }f^{\prime}\left(  0\right)  <0\text{,
}f^{\prime}\left(  1\right)  >0\text{;} \tag{H3}%
\end{equation}

\noindent\label{H5}(H4) the function $g(u):=u+\lambda f\left(  u\right)  $ has
three zeros and exactly three intervals of monotonicity for any sufficiently
large $\lambda>0$. We denote by $u_{\lambda}^{-}$, $u_{\lambda}^{+}$, the
extreme roots of $g(u)=0$, then $u_{\lambda}^{-}<$ $u_{\lambda}^{+}$,
$u_{\lambda}^{-}>-1$, $u_{\lambda}^{+}<1$, and $u_{\lambda}^{\pm}%
\rightarrow\pm1$ as $\lambda\rightarrow\infty$.

The following technical conditions always hold under Hypotheses (H1)-(H4) by
taking $\lambda$ sufficiently large:

\noindent\label{C5}(C5) We take $\alpha^{-}\in\left(  -1,0\right)  $,
$\alpha^{+}\in\left(  0,1\right)  $ such that $f^{\prime}\left(  u\right)
\geq\delta>0$ \ for $u$ in $\left[  -1,\alpha^{-}\right]  \cup\left[
\alpha^{+},1\right]  $. In addition, we assume that $\alpha^{+}$, $\alpha^{-}$
satisfy $u_{\lambda}^{-}<\alpha^{-}<0<\alpha^{+}<u_{\lambda}^{+}$;
\begin{equation}
\left(  1+\alpha^{+}\right)  +\lambda f\left(  \alpha^{+}\right)  \leq0;
\tag{C6}%
\end{equation}%
\begin{equation}
\alpha^{-}+\lambda f\left(  \alpha^{-}\right)  \geq0\text{.} \tag{C7}%
\end{equation}
Notice that conditions C6-C7 hold if $\lambda\geq\max\left\{  \frac
{-\alpha^{-}}{f\left(  \alpha^{-}\right)  },\frac{1+\alpha^{+}}{-f\left(
\alpha^{+}\right)  }\right\}  $.

\begin{remark}
Notice that polynomial $u^{3}-u$\ satisfies hypotheses H1-H4. If $f(u)\in
C^{2}\left(  \mathbb{R}\right)  $ and $f^{\prime\prime}(u)$ has exactly a zero
$\beta\in\left(  -1,1\right)  $ and $f^{\prime\prime}(u)$ is positive to the
right of $\beta$ and negative to the left of $\beta$, then $g(u)=u+\lambda
f\left(  u\right)  $ satisfies Hypothesis H4.
\end{remark}

\begin{theorem}
\label{Theorem1}Assume that $f$ satisfies hypotheses (H1)-(H4). Then for any
measurable subset $M\subset B_{N_{0}}^{n}$ and for $\lambda$ sufficiently
large, the equation%
\begin{equation}
\left\{
\begin{array}
[c]{l}%
u\in X_{\infty}\\
Au\left(  x\right)  +\lambda f\left(  u\left(  x\right)  \right)  =0\text{,
}x\in\mathbb{Q}_{p}^{n}\text{,}%
\end{array}
\right.  \label{Eq2}%
\end{equation}
has a unique solution $\widetilde{u}$ satisfying%
\begin{equation}
\alpha^{+}\leq\widetilde{u}\left(  x\right)  \leq1\text{ for }x\in M\text{
and\ \ }-1\leq\widetilde{u}\left(  x\right)  \leq\alpha^{-}\text{ for }x\in
B_{N_{0}}^{n}\smallsetminus M\text{.} \label{Eq3}%
\end{equation}

\end{theorem}

\begin{proof}
We define
\[
Y:=\left\{  u\in X_{\infty};u\left(  x\right)  \in\left[  \alpha^{+},1\right]
\text{ for }x\in M\text{ and\ }u\left(  x\right)  \in\left[  -1,\alpha
^{-}\right]  \text{ for }x\in B_{N_{0}}^{n}\smallsetminus M\right\}
\]
and the mapping $T:Y\rightarrow X_{\infty}$ as
\begin{align*}
Tu\left(  x\right)   &  =u\left(  x\right)  -h\left\{  Au\left(  x\right)
+\lambda f\left(  u\left(  x\right)  \right)  \right\} \\
&  =\left(  1-h\right)  u\left(  x\right)  +hJ\left(  \left\Vert x\right\Vert
_{p}\right)  \ast u\left(  x\right)  -h\lambda f\left(  u\left(  x\right)
\right)
\end{align*}
for $h>0$. By Lemma \ref{Lemma_E}, $A:Y\rightarrow X_{\infty}$ is well-defined
operator, and $u\rightarrow f\left(  u\right)  $ is also a well-defined
mapping from $Y$ into $X_{\infty}$ because $f\left(  u\right)  $ is a
continuous function satisfying $\lim_{\left\Vert x\right\Vert _{p}%
\rightarrow\infty}f\left(  u\left(  x\right)  \right)  =0$, here we use that
$f\left(  0\right)  =0$.

We show that for $h$ sufficiently small $T$ is a contraction from $Y$ into
itself, and since $Y$ is a Banach space, $T$ has a unique fixed point
$\widetilde{u}$, which is the desired solution.

\textbf{Claim 1.} For $h$ sufficiently small, if $u\left(  x\right)
\in\left[  -1,\alpha^{-}\right]  \cup\left[  \alpha^{+},1\right]  $,\ then
$-1\leq Tu\left(  x\right)  \leq1$.

Indeed, since $f^{\prime}\left(  u\right)  $ is non-negative and continuous on
$\left[  -1,\alpha^{-}\right]  \cup\left[  \alpha^{+},1\right]  $, we can pick
$h$ sufficiently small to get%
\begin{equation}
h\lambda f^{\prime}\left(  u\right)  <1-h\text{ for }u\in\left[  -1,\alpha
^{-}\right]  \cup\left[  \alpha^{+},1\right]  \text{.} \label{Cond_1}%
\end{equation}
Notice that (\ref{Cond_1}) holds if $0<h$ $<\frac{1}{1+\lambda\max
_{u\in\left[  -1,\alpha^{-}\right]  \cup\left[  \alpha^{+},1\right]
}f^{\prime}\left(  u\right)  }$.

From (\ref{Cond_1}) with $u\left(  x\right)  \in\left[  \alpha^{+},1\right]
$, we get $h\lambda\int_{u}^{1}f^{\prime}\left(  u\right)  du\leq\left(
1-h\right)  \int_{u}^{1}du$, which implies that $\left(  1-h\right)
u-h\lambda f\left(  u\right)  \leq1-h$, and thus $Tu\left(  x\right)
\leq\left(  1-h\right)  +hJ\left(  \left\Vert x\right\Vert _{p}\right)  \ast
u\left(  x\right)  \leq1$, since
\begin{equation}
\left\vert J\left(  \left\Vert x\right\Vert _{p}\right)  \ast u\left(
x\right)  \right\vert \leq\left\Vert J\right\Vert _{L^{1}}\left\Vert
u\right\Vert _{\infty}\leq1. \label{Eq1}%
\end{equation}
In the case in which $u\left(  x\right)  \in\left[  -1,\alpha^{-}\right]  $,
we get from (\ref{Cond_1}) that $h\lambda\int_{-1}^{u}f^{\prime}\left(
u\right)  du\leq\left(  1-h\right)  \int_{-1}^{u}du$, which implies that
$\left(  1-h\right)  u-h\lambda f\left(  u\right)  \geq-1+h\geq-1$, and thus
$Tu\left(  x\right)  \geq-1+hJ\left(  \left\Vert x\right\Vert _{p}\right)
\ast u\left(  x\right)  \geq-1$, by using (\ref{Eq1}).

\textbf{Claim 2.} $T:Y\rightarrow Y$.

\textbf{Case} $u\left(  x\right)  \in\left[  \alpha^{+},1\right]  $ for $x\in
M $.

By using that $u-h\left\{  u+\lambda f\left(  u\right)  \right\}  $ is
increasing on $\left[  \alpha^{+},1\right]  $, we have%
\begin{gather*}
Tu\left(  x\right)  \geq\alpha^{+}-h\alpha^{+}-h\lambda f\left(  \alpha
^{+}\right)  +hJ\left(  \left\Vert x\right\Vert _{p}\right)  \ast u\left(
x\right) \\
\geq\alpha^{+}+h\left\{  \alpha^{+}\int\limits_{M}J\left(  \left\Vert
x-y\right\Vert _{p}\right)  d^{n}y-\alpha^{+}-\lambda f\left(  \alpha
^{+}\right)  \right\} \\
\geq\alpha^{+}+h\left\{  \alpha^{+}\int\limits_{M}J\left(  \left\Vert
x-y\right\Vert _{p}\right)  d^{n}y-\alpha^{+}-\lambda f\left(  \alpha
^{+}\right)  -\int\limits_{\mathbb{Q}_{p}^{n}\smallsetminus M}J\left(
\left\Vert x-y\right\Vert _{p}\right)  d^{n}y\right\} \\
=\alpha^{+}-h\left\{  \left(  1+\alpha^{+}\right)  \int\limits_{\mathbb{Q}%
_{p}^{n}\smallsetminus M}J\left(  \left\Vert x-y\right\Vert _{p}\right)
d^{n}y+\lambda f\left(  \alpha^{+}\right)  \right\}  \geq\\
\alpha^{+}-h\left\{  \left(  1+\alpha^{+}\right)  +\lambda f\left(  \alpha
^{+}\right)  \right\}  ,
\end{gather*}
since $\int_{\mathbb{Q}_{p}^{n}\smallsetminus M}J\left(  \left\Vert
x-y\right\Vert _{p}\right)  d^{n}y\leq1$. Now, by Condition C6 and Claim 1,
$\alpha^{+}\leq Tu\left(  x\right)  \leq1$ for $x\in M$.

\textbf{Case} $u\left(  x\right)  \in\left[  -1,\alpha^{-}\right]  $ for $x\in
B_{N_{0}}^{n}\smallsetminus M$.

By using that $u-h\left\{  u+\lambda f\left(  u\right)  \right\}  $ is
increasing on $\left[  -1,\alpha^{-}\right]  $, cf. (\ref{Cond_1}),%
\begin{gather*}
Tu\left(  x\right)  \leq\alpha^{-}-h\alpha^{-}-h\lambda f\left(  \alpha
^{-}\right)  +hJ\left(  \left\Vert x\right\Vert _{p}\right)  \ast1_{B_{N_{0}%
}^{n}\smallsetminus M}\left(  x\right)  u\left(  x\right) \\
\leq\alpha^{-}-h\left\{  \alpha^{-}-\alpha^{-}\int\limits_{B_{N_{0}}%
^{n}\smallsetminus M}J\left(  \left\Vert x-y\right\Vert _{p}\right)
d^{n}y+\lambda f\left(  \alpha^{-}\right)  \right\} \\
\leq\alpha^{-}-h\left\{  \alpha^{-}+\lambda f\left(  \alpha^{-}\right)
\right\}  ,
\end{gather*}
since $h\alpha^{-}\int_{B_{N_{0}}^{n}\smallsetminus M}J\left(  \left\Vert
x-y\right\Vert _{p}\right)  d^{n}y\leq0$. Now by using Condition C7 and Claim
1, $-1\leq Tu\left(  x\right)  \leq\alpha^{-}$ for $x\in B_{N_{0}}%
^{n}\smallsetminus M$.

\textbf{Claim 3. }$T:Y\rightarrow Y$ is a contraction for $h$ sufficiently small.

Take $u$, $v$ in $X$, then
\begin{align*}
\left\Vert Tu-Tv\right\Vert _{\infty}  &  =\left\Vert \left(  1-h\right)
\left(  u-v\right)  -h\lambda\left\{  f(u)-f(v)\right\}  +hJ\ast\left(
u-v\right)  \right\Vert _{\infty}\\
&  =\left\Vert \left(  1-h-h\lambda f^{\prime}\left(  w\right)  \right)
\left(  u-v\right)  +hJ\ast\left(  u-v\right)  \right\Vert _{\infty}\\
&  \leq\left\vert 1-h-h\lambda f^{\prime}\left(  w\right)  \right\vert
\left\Vert u-v\right\Vert _{\infty}+h\left\Vert u-v\right\Vert _{\infty
}\text{,}%
\end{align*}
with $w=au+\left(  1-a\right)  v$ for some $0<a<1$. By (\ref{Cond_1}),
$1-h-h\lambda f^{\prime}\left(  w\right)  >0$ for $h$ sufficiently small, and
$1-h-h\lambda f^{\prime}\left(  w\right)  \leq1$ because $-h\left(  1+\lambda
f^{\prime}\left(  w\right)  \right)  \leq0$ given that $\lambda f^{\prime
}\left(  w\right)  \geq\lambda\delta>0$, cf. Condition C5. Hence
\[
\left\Vert Tu-Tv\right\Vert _{\infty}\leq\left(  1-h-h\lambda f^{\prime
}\left(  w\right)  \right)  \left\Vert u-v\right\Vert _{\infty}\leq\left\Vert
u-v\right\Vert _{\infty}\text{.}%
\]
By using the fact that $Y$ is a Banach space we conclude that the equation
(\ref{Eq2}) has a unique solution $\widetilde{u}$ satisfying (\ref{Eq3}). Now,
from Hypothesis H4 and using that $\widetilde{u}=g_{i}^{-1}\left(
J\ast\widetilde{u}\right)  $, $i=1,2$, where $g_{i}^{-1}$ is one of the
extreme branches of $g^{-1}$, and the fact that $J\ast\widetilde{u}$ is a
continuous function (by the Dominated Convergence Theorem), we conclude that
$\widetilde{u}$ is a continuous function.
\end{proof}

\begin{remark}
In Theorem \ref{Theorem1}, ball $B_{N_{0}}^{n}$ can be replaced by a compact
subset, however, for the sake of simplicity we use a ball centered at the origin.
\end{remark}

\begin{theorem}
\label{Theorem1A}Fix $N_{0}\geq1$ and assume that supp $J\left(  \left\Vert
x\right\Vert _{p}\right)  \subset B_{N_{0}}^{n}$, and that $f$ satisfies
hypotheses (H1)-(H4).Then for any open and compact subset $M$ contained in
$B_{N_{0}}^{n}$ and for $\lambda$ sufficiently large, the equation%
\begin{equation}
\left\{
\begin{array}
[c]{l}%
u\in X_{N_{0}}\\
A_{N_{0}}u\left(  x\right)  +\lambda f\left(  u\left(  x\right)  \right)
=0\text{, }x\in B_{N_{0}}^{n}\text{,}%
\end{array}
\right.  \label{Eq2A}%
\end{equation}
has a unique solution $\widetilde{u}$ satisfying%
\begin{equation}
\alpha^{+}\leq\widetilde{u}\left(  x\right)  \leq1\text{ for }x\in M\text{
and\ \ }-1\leq\widetilde{u}\left(  x\right)  \leq\alpha^{-}\text{ for }x\in
B_{N_{0}}^{n}\smallsetminus M\text{. } \label{Eq3A}%
\end{equation}

\end{theorem}

\begin{proof}
Due to the fact that $X_{N_{0}}$ is a subspace of $X_{\infty}$ the proof of
this result is completely similar to the one given for Theorem \ref{Theorem1}.
We define
\[
Y_{N_{0}}:=\left\{  u\in X_{N_{0}};u\left(  x\right)  \in\left[  \alpha
^{+},1\right]  \text{ for }x\in M\text{ and\ }u\left(  x\right)  \in\left[
-1,\alpha^{-}\right]  \text{ for }x\in B_{N_{0}}^{n}\smallsetminus M\right\}
\]
and the mapping $T_{N_{0}}:Y_{N_{0}}\rightarrow X_{N_{0}}$ as
\begin{align*}
T_{N_{0}}u\left(  x\right)   &  =u\left(  x\right)  -h\left\{  A_{N_{0}%
}u\left(  x\right)  +\lambda f\left(  u\left(  x\right)  \right)  \right\} \\
&  =\left(  1-h\right)  u\left(  x\right)  +hJ\left(  \left\Vert x\right\Vert
_{p}\right)  \ast u\left(  x\right)  -h\lambda f\left(  u\left(  x\right)
\right)  .
\end{align*}
for $h>0$. By Lemma \ref{Lemma_D}, $J\left(  \left\Vert x\right\Vert
_{p}\right)  \ast u\left(  x\right)  $ gives rise to a linear operator from
$X_{N}$ into itself, \ and since%
\begin{equation}
f\left(  u\left(  x\right)  \right)  =\sum_{\boldsymbol{i}\in G_{N_{0}}^{n}%
}f\left(  u\left(  \boldsymbol{i}\right)  \right)  \Omega\left(  p^{N_{0}%
}\left\Vert x-\boldsymbol{i}\right\Vert _{p}\right)  , \label{condition_on_f}%
\end{equation}
we have $T_{N_{0}}u\in X_{N_{0}}$ if $u\in Y_{N_{0}}$. Now the proof continues
as in the proof of Theorem \ref{Theorem1}.
\end{proof}

\begin{remark}
By using the fact that $M$ is open compact, there exists $I_{M}$ a subset of
$G_{N_{0}}^{n}$ such that $M=%
%TCIMACRO{\tbigsqcup \nolimits_{\boldsymbol{i}\in I_{M}}}%
%BeginExpansion
{\textstyle\bigsqcup\nolimits_{\boldsymbol{i}\in I_{M}}}
%EndExpansion
\boldsymbol{i}+\left(  p^{N_{0}}\mathbb{Z}_{p}\right)  ^{n}$. By identifying
$u\left(  x\right)  $ with the column vector $\left[  u\left(  \boldsymbol{i}%
\right)  \right]  _{\boldsymbol{i}\in G_{N_{0}}^{n}}$\ and using
(\ref{condition_on_f}), the equation (\ref{Eq2A}) takes the form
\begin{equation}
A^{\left(  N_{0}\right)  }\left[  u\left(  \boldsymbol{i}\right)  \right]
_{\boldsymbol{i}\in G_{N_{0}}^{n}}-\lambda\left[  f\left(  u\left(
\boldsymbol{i}\right)  \right)  \right]  _{\boldsymbol{i}\in G_{N_{0}}^{n}}=0,
\label{Eq4A}%
\end{equation}
then, Theorem \ref{Theorem1A} asserts that the\ equation (\ref{Eq4A}) has a
unique solution $\left[  \widetilde{u}_{\boldsymbol{i}}\right]
_{\boldsymbol{i}\in G_{N_{0}}^{n}}$ satisfying
\[
\alpha^{+}\leq\widetilde{u}\left(  \boldsymbol{i}\right)  \leq1\text{ for
}\boldsymbol{i}\in I_{M}\text{ and }-1\leq\widetilde{u}\left(  \boldsymbol{i}%
\right)  \leq\alpha^{-}\text{ for }\boldsymbol{i}\in G_{N_{0}}^{n}%
\smallsetminus I_{M}.
\]

\end{remark}

The proofs of Theorems \ref{Theorem1}, \ref{Theorem1A} use the same technique
as the proof of \cite[Theorem 2.1]{BatesChmdj}. However, the hypotheses that
we use are simpler than in \cite[Theorem 2.1]{BatesChmdj}.

\section{A Comparison Theorem}

We consider $\mathbb{Q}_{p}^{n}\times\mathbb{R}$ as a topological space with
the product topology. For $T>0$, we set $\mathbb{D}=B_{L}^{n}\times\left(
0,T\right)  $. Notice that the topological closure of $\mathbb{D}$ is
$\overline{\mathbb{D}}=B_{L}^{n}\times\left[  0,T\right]  $ and that the
boundary of $\mathbb{D}$ is Bd$\mathbb{D}=B_{L}^{n}\times\left\{  0\right\}
\cup B_{L}^{n}\times\left\{  T\right\}  $.

\begin{theorem}
\label{Theorem2}Let $u\left(  x,t\right)  $, $v\left(  x,t\right)
:\mathbb{Q}_{p}^{n}\times\mathbb{R\rightarrow R}$ be functions such that the
following conditions hold for $L\geq L_{0}(T)$ for some $L_{0}(T)\in
\mathbb{N}$:

\noindent(C1) $u\left(  x,t\right)  $, $v\left(  x,t\right)  $ are continuous
functions of $x$ in $B_{L}^{n}$ and continuous differentiable functions of $t
$ on $\left[  0,T\right]  $;

\noindent(C2) with the notation $Pw\left(  x,t\right)  :=\frac{\partial
w\left(  x,t\right)  }{\partial t}-\left\{  J\left(  \left\Vert x\right\Vert
_{p}\right)  \ast w\left(  x,t\right)  -w\left(  x,t\right)  \right\}  $,%
\[
Pu\left(  x,t\right)  +\lambda f\left(  u\left(  x,t\right)  \right)  \geq
Pv\left(  x,t\right)  +\lambda f\left(  v\left(  x,t\right)  \right)  \text{
for }\left(  x,t\right)  \in\mathbb{D}\text{;}%
\]

\noindent(C3) $u\left(  x,0\right)  \geq v\left(  x,0\right)  $ for $x$ in
$B_{L}^{n}$.

Then under the hypotheses (C1)-(C3), it verifies that
\[
u\left(  x,t\right)  \geq v\left(  x,t\right)  \text{ for }\left(  x,t\right)
\in\overline{\mathbb{D}}\text{.}%
\]
In addition, if $\left\Vert u\left(  \cdot,t\right)  -v\left(  \cdot,t\right)
\right\Vert _{\infty}\leq C<\infty$, where $C$ is a constant independent of
$t$, there exists $L_{0}\in\mathbb{N}$ \ such that for $L\geq L_{0}$,%
\[
u\left(  x,t\right)  \geq v\left(  x,t\right)  \text{ for }\left(  x,t\right)
\in B_{L}^{n}\times\left[  0,\infty\right)  .
\]

\end{theorem}

\begin{proof}
Set $w\left(  x,t\right)  :=u\left(  x,t\right)  -v\left(  x,t\right)  $. Then
$w\left(  x,0\right)  \geq0$ for $x$ in $B_{L}^{n}$\ and
\[
Pw\left(  x,t\right)  +\lambda f^{\prime}\left(  w_{0}\right)  w\left(
x,t\right)  \geq0\text{ for }\left(  x,t\right)  \in\mathbb{D}\text{,}%
\]
where $w_{0}=au+\left(  1-a\right)  v$, for some $a:=a(x,t)$ satisfying
$0<a<1$. By contradiction, assume that $w<0$ at some point of $\overline
{\mathbb{D}} $. Set $z\left(  x,t\right)  =e^{-kt}w\left(  x,t\right)  $ with
$k$ a positive constant satisfying%
\[
\lambda f^{\prime}\left(  w_{0}\right)  +k>0\text{ in }\overline{\mathbb{D}%
}\text{.}%
\]
Then
\begin{equation}
Pw+\lambda f^{\prime}\left(  w_{0}\right)  w=e^{kt}\left\{  Pz+z\left(
\lambda f^{\prime}\left(  w_{0}\right)  +k\right)  \right\}  \geq0\text{ in
}\mathbb{D}\text{,} \label{Eq4}%
\end{equation}
and $z\left(  x,0\right)  =w\left(  x,0\right)  \geq0$ in $B_{L}^{n}$. Since
$z$ is negative at some point of $\overline{\mathbb{D}}$,
\[
m:=\min_{\left(  x,t\right)  \in\overline{\mathbb{D}}}z\left(  x,t\right)
<0\text{ and }m=z\left(  x_{0},t_{0}\right)  \text{ for some }\left(
x_{0},t_{0}\right)  \in\overline{\mathbb{D}}\text{.}%
\]
Three cases occur: (i) $\left(  x_{0},t_{0}\right)  \in B_{L}^{n}%
\times\left\{  0\right\}  $, (ii) $\left(  x_{0},t_{0}\right)  \in B_{L}%
^{n}\times\left\{  T\right\}  $, (iii) $\left(  x_{0},t_{0}\right)
\in\mathbb{D}$. The first case is not possible due to condition (C3). We now
consider case (iii). By (\ref{Eq4}) we get
\[
Pz\left(  x_{0},t_{0}\right)  \geq-\left(  \lambda f^{\prime}\left(
w_{0}\right)  +k\right)  z\left(  x_{0},t_{0}\right)
\]
and since $\left(  \lambda f^{\prime}\left(  w_{0}\right)  +k\right)  >0$ in
$\mathbb{D}$, we get $Pz\left(  x_{0},t_{0}\right)  >0$. A contradiction is
derived from Claim 4, and thus this case is not possible.

\textbf{Claim 4. }If\textbf{\ }$\left(  x_{0},t_{0}\right)  \in\mathbb{D}$,
then $Pz\left(  x_{0},t_{0}\right)  \leq0$ for $L$ sufficiently large.

The proof requires to consider two subcases: (i) supp $J\subset B_{L}^{n}$ for
some $L\in\mathbb{N}$, (ii) supp $J\nsubseteqq B_{L}^{n}$ for any
$L\in\mathbb{N}$. In the first subcase, $\frac{\partial z\left(  x_{0}%
,t_{0}\right)  }{\partial t}\leq0$ since $z\left(  x_{0},t\right)  $ has a
global minimum in $\left[  0,T\right]  $, and since%
\[
\left(  J\ast z-z\right)  \left(  x_{0},t_{0}\right)  =\int\limits_{B_{L}^{n}%
}J\left(  \left\Vert x_{0}-y\right\Vert _{p}\right)  \left\{  z\left(
y,t_{0}\right)  -z\left(  x_{0},t_{0}\right)  \right\}  d^{n}y\geq0,
\]
hence $Pz\left(  x_{0},t_{0}\right)  \leq0$. In the second subcase,
$\frac{\partial z\left(  x_{0},t_{0}\right)  }{\partial t}\leq0$ and%
\begin{gather*}
\left(  J\ast z-z\right)  \left(  x_{0},t_{0}\right)  =\int\limits_{\mathbb{Q}%
_{p}^{n}}J\left(  \left\Vert x_{0}-y\right\Vert _{p}\right)  \left\{  z\left(
y,t_{0}\right)  -z\left(  x_{0},t_{0}\right)  \right\}  d^{n}y\\
=\int\limits_{B_{L}^{n}}J\left(  \left\Vert x_{0}-y\right\Vert _{p}\right)
\left\{  z\left(  y,t_{0}\right)  -z\left(  x_{0},t_{0}\right)  \right\}
d^{n}y+\int\limits_{\mathbb{Q}_{p}^{n}\smallsetminus B_{L}^{n}}J\left(
\left\Vert x_{0}-y\right\Vert _{p}\right)  z\left(  y,t_{0}\right)  d^{n}y\\
-z\left(  x_{0},t_{0}\right)  \int\limits_{\mathbb{Q}_{p}^{n}\smallsetminus
B_{L}^{n}}J\left(  \left\Vert x_{0}-y\right\Vert _{p}\right)  d^{n}%
y=:I_{0}+I_{1}+I_{2}.
\end{gather*}
Now, $I_{0}\geq0$ since $z\left(  y,t_{0}\right)  \geq z\left(  x_{0}%
,t_{0}\right)  $ for $y\in B_{L}^{n}$, and $I_{2}>0$ since
\[
\int\limits_{\left\Vert y\right\Vert _{p}>p^{L}}J\left(  \left\Vert
x_{0}-y\right\Vert _{p}\right)  d^{n}y=\int\limits_{\left\Vert y\right\Vert
_{p}>p^{L}}J\left(  \left\Vert y\right\Vert _{p}\right)  d^{n}y\text{ because
}x_{0}\in B_{L}^{n}\text{,}%
\]
this last integral is strictly positive due to the fact that supp
$J\nsubseteqq B_{L}^{n}$ for any $L\in\mathbb{N}$, and that $J$ is a
continuous function. On the other hand,
\[
\left\vert I_{1}\right\vert \leq\left\Vert z\left(  \cdot,t_{0}\right)
\right\Vert _{\infty}\int\limits_{\mathbb{Q}_{p}^{n}\smallsetminus B_{L}^{n}%
}J\left(  \left\Vert x_{0}-y\right\Vert _{p}\right)  d^{n}y\leq C(T)\int
\limits_{\left\Vert y\right\Vert _{p}>p^{L}}J\left(  \left\Vert y\right\Vert
_{p}\right)  d^{n}y,
\]
since $J\left(  \left\Vert y\right\Vert _{p}\right)  \in L^{1}$, for $L$
sufficiently large, $\left\vert I_{1}\right\vert <I_{2}$ and therefore
$Pz\left(  x_{0},t_{0}\right)  <0$ for $L\geq L_{0}(T)$ for some $L_{0}%
(T)\in\mathbb{N}$.

Finally, we consider the case (ii). In this case, since $Pz+z\left(  \lambda
f^{\prime}\left(  w_{0}\right)  +k\right)  \geq0$ in $\mathbb{D}$, by taking
the limit as $\left(  x,t\right)  \rightarrow\left(  x_{0},T\right)  $ we get
$Pz\left(  x_{0},T\right)  \geq0$. The argument given for Claim 4 also works
in the case $\left(  x_{0},T\right)  $ because $z\left(  x_{0},t\right)  $ is
differentiable at $t=T$. Therefore case (ii) is also impossible.
\end{proof}

\begin{lemma}
\label{Lemma1}Set $\overline{u}\left(  x,t\right)  :=\widetilde{u}\left(
x\right)  +\epsilon e^{-\beta t}$, $\underline{u}\left(  x,t\right)
:=\widetilde{u}\left(  x\right)  -\epsilon e^{-\beta t}$, where $\widetilde
{u}\left(  x\right)  $ is the function given in Theorem \ref{Theorem1},
$\epsilon$, $\beta>0$, and $x\in\mathbb{Q}_{p}^{n}$, $t\geq0$. If $\epsilon$,
$\beta$ are sufficiently small, then%
\begin{equation}
P\overline{u}\left(  x,t\right)  +\lambda f\left(  \overline{u}\left(
x,t\right)  \right)  \geq0\text{ \ for }x\in\mathbb{Q}_{p}^{n},t\geq0
\label{Super_sol}%
\end{equation}
and%
\begin{equation}
P\underline{u}\left(  x,t\right)  +\lambda f\left(  \underline{u}\left(
x,t\right)  \right)  \leq0\text{ \ for }x\in\mathbb{Q}_{p}^{n},t\geq0.
\label{Sub_sol}%
\end{equation}

\end{lemma}

\begin{proof}
We show (\ref{Super_sol}), the other inequality is established in the same
way. We first note that%
\[
P\overline{u}\left(  x,t\right)  +\lambda f\left(  \overline{u}\left(
x,t\right)  \right)  =-\beta\epsilon e^{-\beta t}+\lambda\left\{  f\left(
\widetilde{u}\left(  x\right)  +\epsilon e^{-\beta t}\right)  -f\left(
\widetilde{u}\left(  x\right)  \right)  \right\}  .
\]
By using that $\widetilde{u}\left(  x\right)  $ is bounded and the Taylor
Expansion Theorem with $f\in C^{2}\left(  \mathbb{R}\right)  $,
\[
\left\vert f\left(  \widetilde{u}\left(  x\right)  +\epsilon e^{-\beta
t}\right)  -f\left(  \widetilde{u}\left(  x\right)  \right)  -\epsilon
e^{-\beta t}f^{\prime}\left(  \widetilde{u}\left(  x\right)  \right)
\right\vert \leq C\epsilon^{2}e^{-2\beta t},
\]
for some $C>0$. Finally, by (H4) and by choosing $\epsilon$, $\beta$
sufficiently small,%
\[
f\left(  \widetilde{u}\left(  x\right)  +\epsilon e^{-\beta t}\right)  \geq
f\left(  \widetilde{u}\left(  x\right)  \right)  +\epsilon e^{-\beta
t}f^{\prime}\left(  \widetilde{u}\left(  x\right)  \right)  -C\epsilon
^{2}e^{-2\beta t}\geq0.
\]

\end{proof}

\begin{remark}
\label{nota2}If in Lemma \ref{Lemma1}, operator $Pw\left(  x,t\right)  $ is
replaced by
\[
\frac{\partial w\left(  x,t\right)  }{\partial t}-\left\{  J_{N}\left(
\left\Vert x\right\Vert _{p}\right)  \ast w\left(  x,t\right)  -j_{N}w\left(
x,t\right)  \right\}
\]
and take $x\in B_{L}^{n},t\geq0$, with $L$ sufficiently large, then
(\ref{Super_sol})-(\ref{Sub_sol}) are valid for $x\in B_{L}^{n},t\geq0$.
\end{remark}

Let us consider \ the following Cauchy problem:%
\begin{equation}
\left\{
\begin{array}
[c]{l}%
\frac{\partial u\left(  x,t\right)  }{\partial t}=J\left(  \left\Vert
x\right\Vert _{p}\ast u\left(  x,t\right)  \right)  -u\left(  x,t\right)
-\lambda f\left(  u\left(  x,t\right)  \right)  \text{,}\\
u\left(  x,0\right)  =u_{0}\left(  x\right)  ,
\end{array}
\right.  \label{eqA}%
\end{equation}
where $\left(  x,t\right)  \in\mathbb{Q}_{p}^{n}\times\left[  0,\infty\right)
$, and $u\left(  x,t\right)  :\mathbb{Q}_{p}^{n}\times\left[  0,\infty\right)
\rightarrow\mathbb{R}$ is a bounded continuous \ function in $x$ for any fixed
$t\geq0$, which is continuously differentiable in $t\in\left[  0,\infty
\right)  $ for any fixed $x\in\mathbb{Q}_{p}^{n}$.

\begin{corollary}
\label{Cor1}Assume that $u\left(  x,t\right)  $ is a solution of Cauchy
problem (\ref{eqA}) with $u_{0}\left(  x\right)  $ a bounded function
satisfying $\widetilde{u}\left(  x\right)  -\epsilon\leq u_{0}\left(
x\right)  \leq\widetilde{u}\left(  x\right)  +\epsilon$, with $\epsilon$
sufficiently small so that Lemma \ref{Lemma1} holds. Then $\underline
{u}\left(  x,t\right)  \leq u\left(  x,t\right)  \leq\overline{u}\left(
x,t\right)  $ for $\left(  x,t\right)  \in\mathbb{Q}_{p}^{n}\times\left[
0,\infty\right)  $.
\end{corollary}

\begin{proof}
The result follows from Theorem \ref{Theorem2} and Lemma \ref{Lemma1}.
\end{proof}

Lest us consider the following Cauchy problem:%
\begin{equation}
\left\{
\begin{array}
[c]{l}%
\frac{\partial u\left(  x,t\right)  }{\partial t}=J_{N}\left(  \left\Vert
x\right\Vert _{p}\ast u\left(  x,t\right)  \right)  -u\left(  x,t\right)
-\lambda f\left(  u\left(  x,t\right)  \right) \\
u\left(  x,0\right)  =u_{0}\left(  x\right)  ,
\end{array}
\right.  \label{eqB}%
\end{equation}
where $\left(  x,t\right)  \in B_{N}^{n}\times\left[  0,\infty\right)  $, and
$u\left(  x,t\right)  :B_{N}^{n}\times\left[  0,\infty\right)  \rightarrow
\mathbb{R}$ is a continuous\ function in $x$ for any fixed $t\geq0$, which is
continuously differentiable in $t\in\left[  0,\infty\right)  $ for any fixed
$x\in B_{N}^{n}$.

\begin{corollary}
\label{Cor2A} Assume that $N$ is sufficiently large and that $u\left(
x,t\right)  $ is a solution of Cauchy problem (\ref{eqB}) with $u_{0}\left(
x\right)  \in X_{N}$ satisfying $\widetilde{u}\left(  x\right)  -\epsilon\leq
u_{0}\left(  x\right)  \leq\widetilde{u}\left(  x\right)  +\epsilon$, with
$\epsilon$ sufficiently small so that Lemma \ref{Lemma1} and Remark
\ref{nota2} hold. Then $\underline{u}\left(  x,t\right)  \leq u\left(
x,t\right)  \leq\overline{u}\left(  x,t\right)  $ for $\left(  x,t\right)  \in
B_{N}^{n}\times\left[  0,\infty\right)  $.
\end{corollary}

\begin{proof}
The result follows from Theorem \ref{Theorem2} and Remark \ref{nota2}.
\end{proof}

\section{The Cauchy Problem}

\begin{theorem}
\label{Theorem3}Consider the Cauchy problem:%
\begin{equation}
\left\{
\begin{array}
[c]{ll}%
u\left(  x,t\right)  \in C\left(  \left[  0,T\right]  ,X_{\infty}\left(
\mathbb{Q}_{p}^{n}\right)  \right)  \cap C^{1}\left(  \left[  0,T\right]
,X_{\infty}\left(  \mathbb{Q}_{p}^{n}\right)  \right)  , & T>0\\
\frac{\partial u\left(  x,t\right)  }{\partial t}=J\left(  \left\Vert
x\right\Vert _{p}\right)  \ast u\left(  x,t\right)  -u\left(  x,t\right)
-\lambda f\left(  u\left(  x,t\right)  \right)  , & t\in\left[  0,T\right] \\
u\left(  x,0\right)  =u_{0}\left(  x\right)  , &
\end{array}
\right.  \label{Cauchy_problem_2}%
\end{equation}
with $u_{0}\left(  x\right)  \in X_{\infty}\left(  \mathbb{Q}_{p}^{n}\right)
$ satisfying $\widetilde{u}\left(  x\right)  -\epsilon\leq u_{0}\left(
x\right)  \leq\widetilde{u}\left(  x\right)  +\epsilon$, with $\lambda$,
$\widetilde{u}\left(  x\right)  $ as in Theorem \ref{Theorem1}, and\ with
$\epsilon$ sufficiently small so that Lemma \ref{Lemma1} holds. Then, the
initial value problem (\ref{Cauchy_problem_2}) has a unique solution
satisfying $\underline{u}\left(  x,t\right)  \leq u\left(  x,t\right)
\leq\overline{u}\left(  x,t\right)  $ for $\left(  x,t\right)  \in
\mathbb{Q}_{p}^{n}\times\left[  0,\infty\right)  $. In addition, $u\left(
x,t\right)  $ \ satisfies then $\lim_{t\rightarrow\infty}\left\Vert u\left(
x,t\right)  -\widetilde{u}\left(  x\right)  \right\Vert _{\infty}=0$.
\end{theorem}

\begin{proof}
We recall that $-Au\left(  x,t\right)  =J\left(  \left\Vert x\right\Vert
_{p}\right)  \ast u\left(  x,t\right)  -u\left(  x,t\right)  $. By Lemma
\ref{Lemma_E}, $A$ gives rise to a linear bounded operator from $X_{\infty}$
onto itself. On the other hand, $A$ is $m$-dissipative, i.e. there exists
$\nu_{0}>0$ such that for all $h\in X_{\infty}$ there exists a solution $u\in
X_{\infty}$ of $u-\nu_{0}Au=h$, cf. \cite[Proposition 2.2.6]{C-H}. Indeed,
consider the operator%
\[%
\begin{array}
[c]{llll}%
T: & X_{\infty} & \rightarrow & X_{\infty}\\
& u & \rightarrow & -\nu J\ast u+\nu u+h,
\end{array}
\]
with $\nu>0$. \ By Lemma \ref{Lemma_E}, $T$ is well-defined and $\left\vert
Tu-Tv\right\vert \leq2\nu\left\Vert u-v\right\Vert _{\infty}$, thus if
$0<2\nu<1$, $T$ is a contraction and by the Banach Fixed Point Theorem there
exits a unique $u\in X_{\infty}$ such that $Tu=u$, \ which implies that $A$ is
an $m$-dissipative operator on $X_{\infty}$. By the Hille-Yosida-Phillips
Theorem, $-A$ is the generator of a contraction semigroup $e^{-tA}$ on
$X_{\infty}$, see e.g. \cite[Theorem 3.4.4]{C-H}. Now, any solution \ of
(\ref{Cauchy_problem_2}) is \ a solution of the following integral equation:%
\begin{equation}
u\left(  x,t\right)  =e^{-tA}u_{0}\left(  x\right)  +\int_{0}^{t}e^{-\left(
t-s\right)  A}f\left(  u\left(  x,s\right)  \right)  ds\text{ for }t\in\left[
0,T\right]  \text{, }T>0\text{,} \label{Eq_Integral_Eq}%
\end{equation}
cf. \cite[Lemma 4.1.1]{C-H}. By using that $\left\Vert u_{0}\left(  x\right)
\right\Vert _{\infty}\leq\left\Vert \widetilde{u}\left(  x\right)  \right\Vert
_{\infty}+\epsilon\leq1+\epsilon=:M<\infty$, there exists a unique solution
$u\left(  x,t\right)  \in C\left(  \left[  0,T_{M}\right]  ,X_{\infty}\left(
\mathbb{Q}_{p}^{n}\right)  \right)  $ of (\ref{Eq_Integral_Eq}), cf.
\cite[Proposition 4.3.3]{C-H}. Then, two cases occur: (i) $T_{M}=\infty$, i.e.
there exists a global solution for (\ref{Cauchy_problem_2}); (ii)
$T_{M}<\infty$ and $\lim_{t\rightarrow T_{M}}\left\Vert u\left(
\cdot,t\right)  \right\Vert _{\infty}=\infty$, cf. \cite[Theorem 4.3.4]{C-H},
now by Corollary \ref{Cor1}, $\left\Vert u\left(  \cdot,t\right)  \right\Vert
_{\infty}\leq M$, therefore $T_{M}=\infty$. Finally, by using Corollary
\ref{Cor1} and Lemma \ref{Lemma1}, $\lim_{t\rightarrow\infty}\left\Vert
u\left(  x,t\right)  -\widetilde{u}\left(  x\right)  \right\Vert _{\infty}=0$.
\end{proof}

By using the same reasoning we obtain the following finite dimensional version
of Theorem \ref{Theorem3}:

\begin{theorem}
\label{Theorem3A}Consider the Cauchy problem:%
\begin{equation}
\left\{
\begin{array}
[c]{ll}%
u\left(  x,t\right)  \in C\left(  \left[  0,T\right]  ,X_{N}\right)  \cap
C^{1}\left(  \left[  0,T\right]  ,X_{N}\right)  , & T>0\\
\frac{\partial u\left(  x,t\right)  }{\partial t}=J_{N}\left(  \left\Vert
x\right\Vert _{p}\right)  \ast u\left(  x,t\right)  -u\left(  x,t\right)
-\lambda f\left(  u\left(  x,t\right)  \right)  , & x\in B_{N}^{n},t\in\left[
0,T\right] \\
u\left(  x,0\right)  =u_{0}\left(  x\right)  , &
\end{array}
\right.  \label{Cauchy_problem_2A}%
\end{equation}
with $N$ sufficiently large, $u_{0}\left(  x\right)  \in X_{N}$ satisfying
$\widetilde{u}\left(  x\right)  -\epsilon\leq u_{0}\left(  x\right)
\leq\widetilde{u}\left(  x\right)  +\epsilon$, with $\lambda$, $\widetilde
{u}\left(  x\right)  $ as in Theorem \ref{Theorem1A}, and with $\epsilon$
sufficiently small so that Lemma \ref{Lemma1} and Remark \ref{nota2} hold.
Then the initial value problem (\ref{Cauchy_problem_2A}) has a unique solution
satisfying $\underline{u}\left(  x,t\right)  \leq u\left(  x,t\right)
\leq\overline{u}\left(  x,t\right)  $ for $\left(  x,t\right)  \in B_{N}%
^{n}\times\left[  0,\infty\right)  $. In addition, $u\left(  x,t\right)  $
satisfies $\lim_{t\rightarrow\infty}\left\Vert u\left(  x,t\right)
-\widetilde{u}\left(  x\right)  \right\Vert _{\infty}=0$.
\end{theorem}

\begin{remark}
\label{Nota_mild_sol}If in Theorems \ref{Theorem3}, the hypotheses on the
initial conditions are changed to \textquotedblleft with $u_{0}\left(
x\right)  \in X_{\infty}\left(  \mathbb{Q}_{p}^{n}\right)  $ satisfying
$-1\leq u_{0}\left(  x\right)  \leq1$,\textquotedblright\ then there exists a
unique solution satisfying $-1\leq u\left(  x,t\right)  \leq1$ for
$x\in\mathbb{Q}_{p}^{n}$ and $t\geq0$. A similar result is obtained if in
Theorem \ref{Theorem3A}, the hypotheses on the initial conditions are changed
to \textquotedblleft with $N$ sufficiently large, $u_{0}\left(  x\right)  \in
X_{N}$ satisfying $-1\leq u_{0}\left(  x\right)  \leq1$.\textquotedblright
\end{remark}

\section{Finite Approximations}

In this section we study finite approximations to the solutions of
\begin{equation}
\left\{
\begin{array}
[c]{lll}%
\frac{\partial u\left(  x,t\right)  }{\partial t}+Au\left(  x,t\right)
=-\lambda f\left(  u\left(  x,t\right)  \right)  , & x\in%
%TCIMACRO{\U{211a} }%
%BeginExpansion
\mathbb{Q}
%EndExpansion
_{p}^{n}, & t\geq0\\
u\left(  x,0\right)  =u_{0}\left(  x\right)  , &  &
\end{array}
\right.  \label{Cauchy_problem_3}%
\end{equation}
where function $f(u)$ satisfies all the conditions given in Section
\ref{Sect_Staionary_Solutions}.

\begin{lemma}
[Condition B]\label{Lemma4}For $N\geq1$, $A_{N}\in\mathfrak{B}\left(
X_{N},X_{N}\right)  $, in addition,%
\[
\left\Vert e^{-A_{N}t}\right\Vert \leq1\text{ \ for }t\geq0\text{, }%
N\geq1\text{.}%
\]

\end{lemma}

\begin{proof}
We recall that $e^{-tA_{N}}=e^{-tA^{\left(  N\right)  }}$ for $t\geq0$, since
$e^{-tA^{\left(  N\right)  }}=e^{-tSA^{\left(  N\right)  }S^{-1}}$ for any
invertible matrix $S$ and $t\geq0$. Hence
\[
\left\Vert e^{-tA_{N}}\right\Vert =\left\Vert e^{-tA^{\left(  N\right)  }%
}\right\Vert \leq1\text{ for }t\geq0
\]
because by Theorem \ref{Theorem2B} (i), $P^{\left(  N\right)  }\left(
t\right)  =e^{-tA^{\left(  N\right)  }}$ satisfies $P^{\left(  N\right)
}\left(  t\right)  \boldsymbol{1}=\boldsymbol{1}$ for $t\geq0$, where
$\boldsymbol{1}$ the unit vector which is \ a vector having all its entries
equal to one. With $P^{\left(  N\right)  }\left(  t\right)  =\left[
P_{\mathbf{i}\boldsymbol{j}}^{\left(  N\right)  }\left(  t\right)  \right]  ,$
with $P_{\mathbf{i}\boldsymbol{j}}^{\left(  N\right)  }\left(  t\right)
\geq0$, and taking $\varphi\in X_{N}$ satisfying $\left\Vert \varphi
\right\Vert _{\infty}=1$, we have
\[
\left\Vert P^{\left(  N\right)  }\left(  t\right)  \varphi\right\Vert
_{\infty}=\max_{\boldsymbol{i}}\left\vert \sum_{\boldsymbol{j}}%
P_{\boldsymbol{ij}}^{\left(  N\right)  }(t)\varphi\left(  \boldsymbol{j}%
\right)  \right\vert \leq\max_{\boldsymbol{i}}\sum_{\boldsymbol{j}%
}P_{\boldsymbol{ij}}^{\left(  N\right)  }(t)=1
\]
for $t\geq0$.
\end{proof}

\begin{remark}
$A$ is a linear bounded operator on $X_{\infty}$, and since $\mathcal{D}(%
%TCIMACRO{\U{211a} }%
%BeginExpansion
\mathbb{Q}
%EndExpansion
_{p}^{n})$ is dense in $X_{\infty}$, $A$ is completely determined by its
restriction to $\mathcal{D}(%
%TCIMACRO{\U{211a} }%
%BeginExpansion
\mathbb{Q}
%EndExpansion
_{p}^{n})$. For an easy cross-referencing with \cite{Milan}, we say `$A$ is
densely defined linear operator in $X_{\infty}$'.
\end{remark}

\begin{lemma}
[Condition C']\label{Lemma6}$A$ is densely defined linear operator in
$X_{\infty}$, there exists $\lambda_{0}\in\left(  -\infty,0\right)  \cap
\rho\left(  A\right)  $ and
\[
\lim_{N\rightarrow\infty}\left\Vert A_{N}P_{N}\varphi-P_{N}A\varphi\right\Vert
_{\infty}=\lim_{N\rightarrow\infty}\left\Vert E_{N}P_{N}\varphi-\varphi
\right\Vert _{\infty}=0\text{ for all }\varphi\in\mathcal{D}(%
%TCIMACRO{\U{211a} }%
%BeginExpansion
\mathbb{Q}
%EndExpansion
_{p}^{n}).
\]

\end{lemma}

\begin{proof}
The existence of $\lambda_{0}$ follow from Lemma \ref{Lemma_E}. Take
$\varphi\in\mathcal{D}(%
%TCIMACRO{\U{211a} }%
%BeginExpansion
\mathbb{Q}
%EndExpansion
_{p}^{n})\subset X_{\infty}(%
%TCIMACRO{\U{211a} }%
%BeginExpansion
\mathbb{Q}
%EndExpansion
_{p}^{n})$, then $\varphi\in\mathcal{D}_{N_{0}}^{-N_{0}}$ $\subset
\mathcal{D}_{N}^{-N}$, for some $N_{0}\geq1$ and for every $N\geq N_{0}$, in
addition, $A\varphi=A_{N_{0}}\varphi$. Indeed,
\begin{multline*}
A\varphi=-\int\limits_{%
%TCIMACRO{\U{211a} }%
%BeginExpansion
\mathbb{Q}
%EndExpansion
_{p}^{n}}J\left(  \left\Vert x-y\right\Vert _{p}\right)  \left\{
\varphi\left(  y\right)  -\varphi\left(  x\right)  \right\}  d^{n}y\\
=-\int\limits_{%
%TCIMACRO{\U{211a} }%
%BeginExpansion
\mathbb{Q}
%EndExpansion
_{p}^{n}}J\left(  \left\Vert x-y\right\Vert _{p}\right)  \left\{
\Omega\left(  p^{N_{0}}\left\Vert y\right\Vert _{p}\right)  \varphi\left(
y\right)  -\Omega\left(  p^{N_{0}}\left\Vert p\right\Vert _{p}\right)
\varphi\left(  x\right)  \right\}  d^{n}y\\
=-\int\limits_{B_{N_{0}}^{n}}\Omega\left(  p^{N_{0}}\left\Vert x-y\right\Vert
_{p}\right)  J\left(  \left\Vert x-y\right\Vert _{p}\right)  \left\{
\Omega\left(  p^{N_{0}}\left\Vert y\right\Vert _{p}\right)  \varphi\left(
y\right)  \right. \\
\left.  -\Omega\left(  p^{N_{0}}\left\Vert p\right\Vert _{p}\right)
\varphi\left(  x\right)  \right\}  d^{n}y=-\int\limits_{B_{N_{0}}^{n}}%
J_{N_{0}}\left(  \left\Vert x-y\right\Vert _{p}\right)  \left\{
\varphi\left(  y\right)  -\varphi\left(  x\right)  \right\}  d^{n}y=A_{N_{0}%
}\varphi.
\end{multline*}
Now, since $P_{N}\mid_{X_{N_{0}}}=P_{N_{0}}$ for $N\geq N_{0}$, $P_{N}%
A\varphi=P_{N}\left(  A_{N_{0}}\varphi\right)  =P_{N_{0}}\left(  A_{N_{0}%
}\varphi\right)  =A_{N_{0}}\varphi=A\varphi$, and since $A_{N}\mid_{X_{N_{0}}%
}=A_{N_{0}}$ for $N\geq N_{0}$, $A_{N}P_{N}\varphi=A_{N}\left(  P_{N_{0}%
}\varphi\right)  =A_{N_{0}}\left(  P_{N_{0}}\varphi\right)  =A_{N_{0}}%
\varphi=A\varphi$, therefore $\left\Vert A_{N}P_{N}\varphi-P_{N}%
A\varphi\right\Vert _{\infty}=0$ for $N\geq N_{0}$. On the other hand, since
$E_{N}\mid_{X_{N_{0}}}=E_{N_{0}}$ for $N\geq N_{0}$, $E_{N}P_{N}\varphi
=E_{N}\left(  P_{N_{0}}\varphi\right)  =E_{N_{0}}\left(  P_{N_{0}}%
\varphi\right)  =E_{N_{0}}\varphi=\varphi$, which implies that $\left\Vert
E_{N}P_{N}\varphi-\varphi\right\Vert _{\infty}=0$ for $N\geq N_{0}$.
\end{proof}

\begin{corollary}
[Condition C]\label{Cor2}Assume Conditions A and B, then $A$ is densely
defined linear operator in $X_{\infty}$, and there exists $\lambda_{0}%
\in\left(  -\infty,0\right)  \cap\rho\left(  A\right)  $ such that for all
$\varphi$ in a dense subset of $X_{\infty}$,
\[
\lim_{N\rightarrow\infty}\left\Vert E_{N}\left(  A_{N}-\lambda_{0}\right)
^{-1}P_{N}\varphi-\left(  A-\lambda_{0}\right)  ^{-1}\varphi\right\Vert
_{\infty}=0.
\]

\end{corollary}

\begin{proof}
See e.g. Lemma 5.4.1 in \cite{Milan}.
\end{proof}

\subsection{Finite Approximations for $p$-adic reaction-ultradiffusion
equations}

Our goal is to approximate the solution $u\left(  x,t\right)  $ of the Cauchy
Problem (\ref{Cauchy_problem_3}) in $X_{\infty}$ using only that $u_{0}\left(
x\right)  \in X_{\infty}$ and $-1\leq u_{0}\left(  x\right)  \leq1$. The
techniques for constructing such approximations are well-known, here we use
reference \cite[Section 5.4]{Milan}. It is possible to approximate $u(x,t)$
without using any a priori information on the initial solution, however this
requires to impose to the nonlinearity $f$ to be globally Lipschitz, this last
condition reduces significantly the potentials $W$ to which we can apply our results.

The discretization of the Cauchy problem (\ref{Cauchy_problem_3}) in the
spaces $X_{N}$ takes the following form:%
\begin{equation}
\left\{
\begin{array}
[c]{l}%
\frac{d}{dt}u_{N}\left(  t\right)  +A_{N}u_{N}\left(  t\right)  =-\lambda
P_{N}f\left(  E_{N}u_{N}\left(  t\right)  \right) \\
u_{N}\left(  0\right)  =P_{N}u_{0}.
\end{array}
\right.  \label{Cauchy_problem_4}%
\end{equation}
By taking $P_{N}u_{0}\left(  x\right)  =\sum_{\boldsymbol{i}\in G_{N}^{n}%
}u_{0}\left(  \boldsymbol{i}\right)  \Omega\left(  p^{N}\left\Vert
x-\boldsymbol{i}\right\Vert _{p}\right)  $ and identifying $u_{N}\left(
t\right)  $ with the column vector $\left[  u_{N}\left(  \boldsymbol{i}%
,t\right)  \right]  _{\boldsymbol{i}\in G_{N}^{n}}$, we can rewrite the Cauchy
problem (\ref{Cauchy_problem_4}) as
\begin{equation}
\left\{
\begin{array}
[c]{l}%
\frac{d}{dt}\left[  u_{N}\left(  \boldsymbol{i},t\right)  \right]
_{\boldsymbol{i}\in G_{N}^{n}}+A^{\left(  N\right)  }\left[  u_{N}\left(
\boldsymbol{i},t\right)  \right]  _{I\in G_{N}^{n}}=-\lambda\left[  f\left(
u_{N}\left(  \boldsymbol{i},t\right)  \right)  \right]  _{\boldsymbol{i}\in
G_{N}^{n}}\\
\left[  u_{N}\left(  \boldsymbol{i},0\right)  \right]  _{\boldsymbol{i}\in
G_{N}^{n}}=\left[  u_{0}\left(  \boldsymbol{i}\right)  \right]
_{\boldsymbol{i}\in G_{N}^{n}},
\end{array}
\right.  \label{Cauchy_problem_5}%
\end{equation}
cf. Lemma \ref{Lemma_C}.

\begin{theorem}
\label{Theorem4}(i) $-A$ is the generator of a strongly continuous semigroup
$\left\{  e^{-tA}\right\}  _{t\geq0}$ on $X_{\infty}$. Moreover, $\left\Vert
e^{-tA}\right\Vert \leq1$ for $t\geq0$ and
\[
\lim_{N\rightarrow\infty}\sup_{t\geq0}e^{bt}\left\Vert E_{N}e^{-A_{N}t}%
P_{N}\varphi-e^{-tA}\varphi\right\Vert _{\infty}=0\text{ for all }\varphi\in
X_{\infty}\text{, }b\in\left(  0,\infty\right)  .
\]

(ii) Take $u_{0}\left(  x\right)  \in X$ with $-1\leq u_{0}\left(  x\right)
\leq1$. Let $u$ be the solution of (\ref{Cauchy_problem_3}) and let $u_{N}$ be
the solution of (\ref{Cauchy_problem_4}). Then
\[
\lim_{N\rightarrow\infty}\sup_{0\leq t\leq T}\left\Vert E_{N}u_{N}\left(
t\right)  -u\left(  t\right)  \right\Vert _{\infty}=0.
\]

\end{theorem}

\begin{proof}
The first part follows from Conditions A, B, C by using Theorem 5.4.2 in
\cite{Milan}. The proof of the second part is based on the estimation of the
$\left\Vert \cdot\right\Vert _{\infty}$-norm of the difference of a mild
solution of (\ref{Cauchy_problem_3}) and a mild solution of
(\ref{Cauchy_problem_4}). The required estimation follows from conditions A,
B, C, D by using Theorem 5.4.7 in \cite{Milan}.
\end{proof}

\textbf{Acknowledgement.} The author wishes to thank to Sergii Torba, the
editors of Nonlinearity, and the referees for many useful comments and
discussions, which led to an improvement of this work. In addition, the author
thanks Luis Gorostiza for telling him about reference \cite{D-Ma-SM}.

\end{document}